\documentclass[12pt]{article}
\usepackage{graphicx} 
\usepackage[english]{babel}
\usepackage{amssymb, amsmath}
\usepackage{amsthm}
\usepackage{eucal}

\usepackage{color}

\usepackage{subfigure}

\usepackage{latexsym}
\usepackage{multicol}
\usepackage{graphicx}
\usepackage{color}
\usepackage{appendix}
\usepackage{dsfont}

\usepackage{bbold}
\usepackage{bm}

\usepackage{amsmath}
\usepackage{graphicx}

\usepackage{algorithm}
\usepackage{algorithmic}
\usepackage[colorlinks=true, allcolors=blue]{hyperref}

\def\tr{{\rm Tr \,}}

\def\N{{\mathbb N}}

\def\R{{\mathbb R}}

\def\cH{{\cal H}}

\def\cI{{\cal I}}

\def\un{\mathds{1}}

\def\un{\mathds{1}}

\def\dd{\mathrm{d}}

\newtheorem{theorem}{Theorem}
\newtheorem{corollary}[theorem]{Corollary}
\newtheorem{lemma}[theorem]{Lemma}
\newtheorem{proposition}[theorem]{Proposition}
\newtheorem{remark}[theorem]{Remark}

\title{A sparse approximation of the Lieb functional with moment constraints}
\author{Virginie Ehrlacher\thanks{CERMICS, Ecole Nationale des Ponts et Chauss\'ees \& INRIA, virginie.ehrlacher@enpc.fr} \;  and\;  Luca Nenna\thanks{Universit\'e Paris-Saclay, CNRS, Laboratoire de math\'ematiques d'Orsay, 91405, Orsay, France. luca.nenna@universite-paris-saclay.fr}}
\date{\today}

\begin{document}

\maketitle

\paragraph{Conflict of interest statement:} The authors have no competing interests to declare that are relevant to the content of this article.

\paragraph{Data availability statement:} Data sharing not applicable to this article as no datasets were generated or analysed during the current study.

\begin{abstract}
    The aim of this paper is to present new sparsity results about the so-called Lieb functional, which is a key quantity in Density Functional Theory for electronic structure calculations of molecules. The Lieb functional was actually shown by Lieb to be a convexification of the so-called Lévy-Lieb functional. Given an electronic density for a system of $N$ electrons, which may be seen as a probability density on $\R^3$, the value of the Lieb functional for this density is defined as the solution of a quantum multi-marginal optimal transport problem, which reads as a minimization problem defined on the set of trace-class operators acting on the space of electronic wave-functions that are anti-symmetric $L^2$ functions of $\R^{3N}$, with partial trace equal to the prescribed electronic density. We introduce a relaxation of this quantum optimal transport problem where the full partial trace constraint is replaced by a finite number of moment constraints on the partial trace of the set of operators. We show that, under mild assumptions on the electronic density, there exist sparse minimizers to the resulting moment constrained approximation of the Lieb (MCAL) functional that read as operators with rank at most equal to the number of moment constraints. We also prove under appropriate assumptions on the set of moment functions that the value of the MCAL functional converges to the value of the exact Lieb functional as the number of moments go to infinity. We also prove some rates of convergence on the associated approximation of the ground state energy. We finally study the mathematical properties of the associated dual problem and introduce a suitable numerical algorithm in order to solve some simple toy models.
\end{abstract}

\tableofcontents

\bigskip

\section{Introduction}
    The so-called Hohenberg-Kohn or L\'evy-Lieb functional plays a fundamental role in Density Functional Theory for electronic structure calculations. For the sake of simplicity, we use here atomic units and neglect the effect of spin in this work. For a given electronic density $\rho \in L^1(\mathbb{R}^3)$, which we assume here to be of integral equal to $1$ for the sake of simplicity, and a given number of electrons $N\in \mathbb{N}^*$, the L\'evy-Lieb functional $F_{LL}(\rho)$ reads as the solution of the following a minimization problem of the form: 
    $$
    \boxed{F_{LL}[\rho]:= \mathop{\inf}_{\substack{\Psi\in\cH_1^N\\\rho_\Psi=\rho}} \frac{1}{2}\int_{{\R}^{3N}} |\nabla \Psi|^2 + \int_{\mathbb{R}^{3N}} V |\Psi|^2,} 
    $$
    
where  
\begin{itemize}
    \item[(i)] $\mathcal H_1^N:= \bigwedge_{i=1}^N H^1(\mathbb{R}^3)$ is the set of admissible electronic wavefunctions for a system of $N$ electrons with finite kinetic energy, that is the set of antisymmetric functions of $H^1(\mathbb{R}^{3N})$;
    \item[(ii)] for any $\Psi \in \mathcal H_1^N$ and $x\in \mathbb{R}^3$, $\rho_\Psi$ is the electronic density associated to the wavefunction $\psi$;
    \item[(iii)] the function $V : (\mathbb{R}^{3})^N \to \mathbb{R}_+ \cup \{ + \infty \}$ is the electron-electron Coulomb interaction potential.
\end{itemize}
 There is a wide zoology of electronic structure calculation models which rely on various types of approximations of this L\'evy-Lieb functional. Recently, Strictly Correlated Electrons (SCE) based approximation of this functional have drawn an increasing interest from mathematicians because it  gives rise to a symmetric multi-marginal optimal transport problem with Coulomb cost, with the number of marginal constraints being equal to the number of electrons in the system.
 The literature about  the SCE approximation (namely the multi-marginal optimal transport with Coulomb cost) is growing considerably. Recent developments include results on the existence and non-existence of Monge-type solutions  (e.g., \cite{ColMar-15,ColPasMar-15,CotFriKlu-13,Friesecke-19,ButPasGor-12,ColStr-M3AS-15,BinDepKau-arxiv-20}), structural properties of Kantorovich potentials (e.g., \cite{ColMarStr-19, MarGerNen-17,GerKauRaj-ESAIMCOCV-19,ButChaPas-17}), grand-canonical optimal transport \cite{MarLewNen-19_ppt}, efficient computational algorithms (e.g., \cite{benamou2017numerical,friesecke2022genetic, CoyEhrLom-19, MarGer-19_ppt, KhoLinLinLex-19_ppt}) and the design of new density functionals (e.g., \cite{GerGroGor-JCTC-20, CheFri-MMS-15, MirUmrMorGor-JCP-14, LanMarGerLeeGor-PCCP-16}). 
 
\noindent Moreover, recent works indicate that the solution of this symmetric Coulomb cost multi-marginal problem (MMOT), which is a probability measure on $\mathbb{R}^{3N}$, is actually a sparse object at least in discrete settings. Two types of discrete settings have been considered so far where such sparsity results have been obtained. On the one hand, the most classical discrete approximation consists in introducing a discrete grid $\mathcal X$ of $\mathbb{R}^3$. The discrete optimal transport plan is then defined as a discrete probability measure defined on the cartesian product grid $\mathcal X^N$. Actually, it was proved in~\cite{friesecke2018breaking,vogler2021geometry} that the discrete optimal transport plan does not charge all the points of the discrete cartesian product grid (of cardinality $|\mathcal X|^N$) but only a number of points in this grid which scales at most linearly with $M$. Finding the few points of $\mathcal X^N$ which are actually charged by the discrete optimal transport plan is not a trivial task though, and the GenCol algorithm is a numerical procedure which aims at achieving this task. It has been first proposed in~\cite{friesecke2022genetic}, then extended in~\cite{friesecke2022gencol} and its convergence has been analyzed for two-marginal problems in~\cite{friesecke2023convergence}. On the other hand, an alternative approach which was first considered in~\cite{alfonsi2021approximation} consists in introducing an approximation of the exact multi-marginal transport problems where the marginal constraints are replaced by a finite number of moment constraints associated to a finite number $M$ of "moment functions" which are real-valued functions defined on $\mathbb{R}^3$. Under some natural assumptions, this approximate problem is then equivalent to approximating the solution of the dual problem associated to the exact optimal transport problem, namely the so-called Kantorovich potential, as a linear combination of these moment functions.  The solution of this moment-contrained optimal transport problem is still a probability measure defined on $\mathbb{R}^{3N}$ but is also a sparse object in the sense that it can be written as a discrete measure charging a number of points belonging to $\mathbb{R}^{3N}$ which scales at most linearly with the number of moment constraints. Finding the location of these points then reads as a non-convex optimization problem defined on a continuous (and not a discrete set) set, and stochastic gradient algorithms have been proposed in~\cite{alfonsi2022constrained} in order to find such optimal points, and numerically tested on three-dimensional settings involving $N = 100$ electrons. We also refer the reader to the works~\cite{chen2014numerical,benamou2017numerical,nenna2022ode,lelotte2022external,hu2023global} where alternative numerical methods have been proposed for the computation of the SCE limit of the L\'evy-Lieb functional, which do not rely on sparsity arguments.   

\medskip

The objective of this work is to prove similar sparsity results for the so-called Lieb functional, which is a convex relaxation of the L\'evy-Lieb functional, the expression of which is given under the following form: 
\begin{equation}\label{eq:Lieb}
\boxed{F_L[\rho]:= \mathop{\inf}_{\Gamma \in \mathfrak{S}_1^+(\mathcal H_0^N), \; \rho_\Gamma= \rho} {\rm Tr}\left[ \bigg(-\frac{1}{2}\Delta + V) \Gamma\bigg)\right],}
\end{equation}
where $\displaystyle\mathcal H_0^N:= \bigwedge_{i=1}^N L^2(\mathbb{R}^3)$, $\mathfrak{S}_1^+(\mathcal H_0^N)$ denotes the set of non-negative trace-class self-adjoint operators on $\mathcal H_0^N$ and where $\rho_\Gamma$ is the electronic density associated to $\Gamma \in \mathfrak{S}_1^+(\mathcal H_0^N)$, the precise definition of which will be given below. Actually, problem~\eqref{eq:Lieb} is a particular instance of \itshape quantum optimal transport \normalfont problem. We refer the reader to~\cite{golse2016mean,golse2017schrodinger} for references on earlier works on closely related types of problems.  Notice that problem~\eqref{eq:Lieb} can be understood  as a {\bf quantum version}  of a multi-marginal optimal transport problem. Moreover, it  still enjoys the nice property, as the original problem, of being a linear programming problem.
Our aim here is to prove that solutions of approximations of problems~\eqref{eq:Lieb} where the partial trace constraint is relaxed by a finite number of moment constraints enjoy similar sparsity properties than solutions of moment constrained multi-marginal symmetric classical optimal transport problems, such as those which were established in~\cite{alfonsi2021approximation}. More precisely, we prove, using the so-called Tchakaloff's theorem (notice that for the usual entropic regularization of MMOT we cannot use this kind of approach), that the solutions of moment constrained approximations of~\eqref{eq:Lieb} can be written under the form $\Gamma = \sum_{k=1}^K \alpha_k |\Psi_k \rangle \langle \Psi_k|$, where $K\in \mathbb{N}^*$ scales at most linearly with the number of moment constraints, and where for all $1\leq k \leq K$, $\alpha_k \in [0,1]$, $\Psi_k \in \mathcal H_1^N$ and $|\Psi_k \rangle \langle \Psi_k|$ is the orthogonal projector of $\mathcal H_0^N$ onto the vectorial space spanned by $\Psi_k$ (using bra-ket notation).
We will, finally, exploit this sparsity structure in order to propose some numerical scheme in order to approximate the solution of~\eqref{eq:Lieb}.
Notice that solving \eqref{min:lieb} for a small systems can be exploited (as done in some recent works for the Levy-Lieb functional \cite{shao2023machine,bai2022machine}) in order to build approximations of the Lieb functional for larger systems by means of machine learning techniques.
Let us finally mention here that particular moment-constrained approximations of the Lieb functional have already been considered in~\cite{garrigue2022building} for the construction of Kohn-Sham potentials. The novel results brought by the present contribution in comparison to the latter work is (i) the extension of existence and convergence results to more general moment constraints that the one considered in~\cite{garrigue2022building}; (ii) the results on the sparsity structure of associated minimizers; (iii) convergence rate of the approximate ground state energy; and (iv) study of the mathematical properties of the associated dual problem.    

The outline of the article is the following. In Section~\ref{sec:Lieb}, we recall some fundamental results about the exact Lieb functional. The moment-constrained approximation we consider here and the associated sparsity result on their minimizers is presented in Section~\ref{sec:moment}. Convergence results of the moment-constrained approximation towards the exact Lieb functional are presented in Section~\ref{sec:convergenceLieb}.  In Section~\ref{sec:convergenceGround}, we also prove some rates of convergence of the associated approximation of the ground state energy to the exact one. Section~\ref{sec:duality} is devoted to present some results about the dual formulation of the moment-constrained problem in the case of electronic density with support included in bounded domains. We, finally, introduce a numerical method in Section ~\ref{sec:algo} exploiting the sparsity result and the convenient dual formulation as an SDP problem. Some numerical experiments for small systems are then predented.

\section{The exact Lieb functional}\label{sec:Lieb}
Let us first introduce some notation together with the problem we consider in this work. We use here atomic units and neglect the influence of spin for the sake of simplicity. 

Let $N\in \N^*$ denote the number of electrons in the molecule of interest. Let us assume that there are $N_{\rm nu}\in \N^*$ nuclei in the molecule, the positions and electric charges of which are denoted by $R_1, \ldots, R_{N_{\rm nu}}\in \R^3$ and $Z_1, \ldots, Z_{N_{\rm nu}}\in \N^*$. For all $x\in \R^3$, let us denote by 
$$
v_{\rm nu}(x):= - \sum_{n=1}^{N_{\rm nu}} \frac{Z_n}{|x - R_n|}
$$
the Coulomb electric potential generated at $x\in \R^3$ by the $N_{\rm nu}$ nuclei.

Let $\cH:= H^1(\R^3)$ and $\cH^N:= \bigwedge_{i=1}^N H^1(\R^3)$. For any $\Psi\in \cH^N$, we denote by $\|\Psi\|$ its $L^2(\mathbb{R}^{3N})$ norm and by $\rho_{\Psi}$ the electronic density associated to the wavefunction $\Psi$, namely the real-valued function defined over $\mathbb{R}^3$ as follows: 
$$
\forall x\in \mathbb{R}^3, \quad \rho_{\Psi}(x):= N \int_{(\mathbb{R}^3)^{N-1}} |\Psi(x, x_2, \ldots, x_N)|^2\,dx_2\ldots\,dx_N.
$$

For a given set of nuclei positions ${\bm R}:=(R_1, \ldots, R_{N_{\rm nu}})$ and charges ${\bm Z}:=(Z_1, \ldots, Z_{N_{
\rm nu}})$, one can compute the ground state energy as a minimization over a density $\rho$, that is
\begin{equation}
\label{min:groundstate}
E[{\bm R}, {\bm Z}]=\inf_{\substack{\rho\in\cI^N}}\left\{F_{LL}[\rho]+\int_{\R^3}v_{\rm nu}\rho\right\}, 
\end{equation}
where $\cI^N:=\{\rho\in L^1(\mathbb{R}^3),\;\rho\geq 0,\;\sqrt{\rho}\in H^1(\R^3),\;\int_{\mathbb{R}^3}\rho=N\}$ and
\begin{equation}
    \label{min:levylieb}
    F_{LL}[\rho]:=\inf_{\substack{\Psi\in\cH_1^N\\\rho_\Psi=\rho}}\left\{\frac{1}{2}\int_{\R^{3N}}|\nabla\Psi|^2+\int_{\R^{3N}}V|\Psi|^2\right\}
\end{equation}
is called the Levy-Lieb functional. In (\eqref{min:levylieb}), the function $V: (\mathbb{R}^{3})^N \to \mathbb{R}_+ \cup \{ + \infty\}$ is defined as follows: for all $(x_1, \ldots, x_N)\in (\mathbb{R}^3)^N$, 
\begin{equation}\label{eq:defV}
V(x_1, \ldots, x_N)= \sum_{1\leq i<j\leq N}\frac{1}{|x_i-x_j|}.
\end{equation}
The Levy-Lieb functional is the central object in Density Functional Theory and its knowledge would allow the computation the electronic ground state energy of any molecule. However, it turns out that $F_{LL}$ is not convex, it is therefore convenient to look at a convexification proposed by Lieb \cite{Lieb-83b} where the minimization is performed over the set of mixed states instead of the set of pure ones as in \eqref{min:levylieb}. More precisely, we consider here the alternative minimization problem 
\begin{equation}
    \label{min:lieb}
    F_{L}[\rho]:=\inf_{\substack{\Gamma \in \mathfrak{S}_1^+(\mathcal H_0^N)\\ \rho_{\Gamma}=\rho}}\tr(H_N\Gamma),
\end{equation}
where $H_N:=-\frac{1}{2}\Delta +V$ is a self-adjoint operator on $\mathcal H_0^N$ with domain $D(H_N) = \mathcal H_2^N:= \bigwedge_{i=1}^N H^2(\mathbb{R}^3)$, $\mathfrak{S}_1^+(\cH_0^N)$ denotes the set of trace-class self-adjoint non-negative operators on $\cH_0^N$. For all $\Gamma \in \mathfrak{S}_1^+(\cH_0^N)$, there exists an orthonormal basis $(\Psi_i)_{i\in \mathbb{N}^*}$ of $\cH_0^N$ and a non-increasing sequence $(\alpha_i)_{i\in \mathbb{N}^*}$ of non-negative numbers such that
\begin{equation}\label{eq:eigen}
\Gamma = \sum_{i=1}^{+\infty}\alpha_i |\Psi_i \rangle \langle \Psi_i|,
\end{equation}
using so-called bra-ket notation. Then, the associated electronic density $\rho_\Gamma$ is defined as follows: for all $x\in \mathbb{R}^3$, 
$$
\rho_\Gamma(x):= N\sum_{i=1}^{+\infty} \alpha_i \int_{(\mathbb{R}^{3})^{N-1}}  |\Psi_i(x, x_2, \ldots, x_N)|^2\,dx_2 \ldots\,dx_N = \sum_{i=1}^{+\infty} \alpha_i \rho_{\Psi_i}(x).
$$

We know that there exist positive constants $\varepsilon,D>0$ such that
$H_N + D \geq \varepsilon (-\Delta + {\rm Id})$ (in the sense of self-
adjoint operators on $\mathcal H_0^N$). We also denote by
$\mathfrak{S}_{1,1}(\mathcal H_0^N)$ the set of self-adjoint operators 
$\Gamma$ on $\mathcal H_0^N$ with finite kinetic energy, i.e. such that
${\rm Tr}\left( |H_N + D|^{1/2} \Gamma  |H_N + D|^{1/2}\right) < +\infty$.
\begin{remark}
 It can then be easily checked that, $\Gamma\in \mathfrak{S}_{1,1}(\mathcal
H_0^N)$ if and only if $ \Gamma\in \mathfrak{S}_{1}(\mathcal
H_0^N)$ and ${\rm Tr}(H_N\Gamma)<+\infty$. Then, if $\Gamma$ admits an
eigendecomposition of the form (\ref{eq:eigen}), necessarily $\Psi_i \in
\mathcal H_1^N$ as soon as $\alpha_i >0$.    
\end{remark}

It is well-known then that the infimum in \eqref{min:levylieb} and \eqref{min:lieb} is attained.
\begin{remark}[Convexification]
It is worth highlighting that $F_L$ is indeed the convexification of $F_{LL}$ in the sense that
\[ F_{L}[\rho]=\inf_{\substack{\forall i\geq 1, \; \alpha_i\geq 0, \; \rho_i\in \mathcal I^N\\
\sum_{i=1}^{+\infty} \alpha_i=1\\  
\sum_{i=1}^{+\infty}\alpha_i\rho_i= \rho\\ }}\sum_{i=1}^{+\infty}\alpha_iF_{LL} [\rho_i]  \]
\end{remark}
It is useful noticing that $F_L$ admits a dual problem.
\begin{theorem}[\cite{Lieb-83b}]
Duality holds in the sense that
\begin{equation}
    \label{eq:duality}
    F_{L}[\rho]=\sup_{\substack{v\in L^{3/2}(\R^3)+L^{\infty}(\R^3)\\H^{v}_N\geq0}}\left\{\int_{\R^3}v(x)\rho(x)\dd x \right\},
\end{equation}
where
\[H^{v}_N=H_N-\sum_{i=1}^Nv(x_i).\]
\end{theorem}
The constraint in \eqref{eq:duality} has to be understood in the sense of self-adjoint operators, namely for all $\Psi \in \mathcal H_1^N$, $\langle\Psi|H^{v}_N|\Psi\rangle \geq 0$.
\begin{remark}
    It is important to notice, for the following, that it can be easily proved that the infimum in \eqref{min:levylieb} and \eqref{min:lieb} is attained. However, it happens that the supremum in \eqref{eq:duality} is not attained for most densities $\rho$ (we refer the reader to \cite{LewLiebSeir-2019}). 
\end{remark}

\section{Moment-constrained approximation and sparsity result}\label{sec:moment}
We focus now on a first approximation of \eqref{min:lieb} by using the moment constraint approach which has previously been studied in the framework of classical optimal transport \cite{alfonsi2021approximation,alfonsi2022constrained}. We also refer to~\cite{garrigue2022building} where a particular instance of moment-constrained approximation of the Lieb functional has been considered for the computation of Kohn-Sham potentials. 

\medskip

We begin by introducing here some notation. From now on, we fix an electronic density $\rho \in \mathcal I_N$. Let us recall that we have $\mathcal F:= L^{3/2}(\mathbb{R}^3) + L^\infty(\mathbb{R}^3) \subset L^1_\rho(\mathbb{R}^3)$. For any $f\in \mathcal F$, we denote by
$$
\|f\|_{\mathcal F}:= \mathop{\inf}_{
\begin{array}{c}
f_{3/2}\in L^{3/2}(\mathbb{R}^3), \; f_\infty \in L^\infty(\mathbb{R}^3),\\
f_{3/2} + f_\infty = f\\
\end{array}}
\|f_{3/2}\|_{L^{3/2}(\mathbb{R}^3)} + \|f_{\infty}\|_{L^{\infty}(\mathbb{R}^3)} .
$$

Let $M\in \mathbb{N}^*$, given a collection of $M$ functions $\Phi:=(\varphi_1,\ldots,\varphi_M)\in \mathcal F^M$, the main idea of the moment-constrained approximation consists in replacing the density constraint in \eqref{min:lieb} with the $M$ scalar moment constraints associated to the functions $\varphi_1,\ldots,\varphi_M$, that is
\begin{equation}\label{eq:moment}
\int_{\mathbb{R}^3} \varphi_m\rho_\Gamma=\int_{\mathbb{R}^3}\varphi_m\rho,\quad\forall m=1,\cdots,M.
\end{equation}
Notice that (\ref{eq:moment}) is equivalent to 
\begin{equation}\label{eq:momenttotal}
\int_{\mathbb{R}^3} \varphi \rho_\Gamma = \int_{\mathbb{R}^3} \varphi \rho, \quad \forall \varphi \in {\rm Span}\{ \Phi \}. 
\end{equation}
We denote by $\mathfrak{S}_1^+(\mathcal H_0^N,\Phi,\rho)$ the set of $\Gamma \in \mathfrak{S}_1^+(\mathcal H_0^N)$ satisfying constraints (\ref{eq:moment}) (or equivalently (\ref{eq:momenttotal})).

In the following, we show that there exists at least one solution to the corresponding moment-constrained Lieb optimization problem admits a sparse solution $\Gamma_{\rm opt}^\Phi$, such that there exists an integer $ K\leq M+2$, weights $\omega_1,\cdots \omega_K\geq0$ and wavefunctions $\Psi_1,\cdots,\Psi_K\in\cH_1^N$ such that 
\begin{equation}\label{eq:sparse}
\sum_{k=1}^K \omega_k=1\quad\text{and}\quad
\Gamma_{\rm opt}^\Phi=\sum_{k=1}^K\omega_k|\Psi_k\rangle\langle\Psi_k|.
\end{equation}
In other words, we will show that there exists a finite-rank minimizer $\Gamma_{\rm opt}^\Phi$ the rank of which is at most $K\leq M+2$.

\subsection{Tchakaloff's theorem on Hilbert spaces}
Let us first recall the following proposition which is an immediate consequence of Tchakaloff's theorem, see \cite{bayer2006proof}.  For any Hilbert space $\mathcal H$, we denote by $\mathcal B(\mathcal H)$ the Borel $\sigma$-algebra of $\mathcal H$. 

\begin{proposition}
\label{prop:tchakaloff}
Let $\mu$ be a Borelian measure on a Hilbert space $\mathcal H$ concentrated on a Borel set $\mathcal A\in\mathcal B(\mathcal H)$. Let $J_0\in\N^*$ and $\Lambda:\mathcal H\to\R^{J_0}$ a Borel measurable map. Assume that the first moments of $\Lambda_{\sharp}\mu$ exists, that is
\[ \int_{\R^{J_0}}\|x\|\dd\Lambda_{\sharp}\mu(x)=\int_{\mathcal H}\|\Lambda(\Psi)\|\dd\mu(\Psi)<+\infty,\]
where $\|\cdot\|$ denotes the Euclidean norm of $\R^{J_0}$. Then there exists an integer $1\leq K\leq J_0$, elements $\Psi_1,\cdots,\Psi_K\in \mathcal A$ and weights $\omega_1,\cdots,\omega_K >0$ such that
\[\forall j=1,\cdots,J_0,\;\int_{\mathcal H}\Lambda_j(\Psi)\dd\mu(\Psi)=\sum_{k=1}^K\omega_k\Lambda_j(\Psi_k) = \int_{\mathcal H} \Lambda_i(\Psi) \,d \mu_d(\Psi) ,\]
where $\Lambda_j$ is the $j-$th component of $\Lambda$, and 
$
\displaystyle \mu_d = \sum_{k=1}^K \omega_k \delta_{\Psi_k}
$.
\end{proposition}

The main idea of the proof of the sparsity result announced above is to define a measure associated to an operator $\Gamma\in \mathfrak{S}_1^+(\mathcal H_0^N)$. Assume that the operator $\Gamma$ can be written as 
\begin{equation}\label{eq:decomposition}
\Gamma=\sum_{i=1}^{+\infty} \alpha_i|\Psi_i\rangle\langle\Psi_i|
\end{equation}
for some sequence $(\Psi_i)_{i\in \mathbb{N}^*}$ of normalized functions of $\mathcal H_0^N$ and non-negative real numbers $(\alpha_i)_{i\in\mathbb{N}^*}$ such that $\sum_{i\in \mathbb{N}^*} \alpha_i = N$. Then we can define a Borelian measure $\mu_{\Gamma}:\mathcal B(\mathcal H_0^N)\to\R_+$ associated to the decomposition (\ref{eq:decomposition}) of the operator $\Gamma$ as 
\[\mu_{\Gamma}=\sum_{i=1}^{+\infty}\alpha_i\delta_{\Psi_i}.\]
Naturally, there is no unique such measure $\mu_\Gamma$ associated with an operator $\Gamma$ since it heavily depends on the decomposition (\ref{eq:decomposition}). However, we will see in the following that this is not a problem for our purpose here.





\subsection{Existence of sparse minimizers for Moment Constrained Approximation of Lieb (MCAL) functional}\label{sec:sparsity}

In the following, we denote by $\un$ the function defined over $\mathbb{R}^3$ which is identically equal to $1$.

We then have the following theorem, the proof of which is postponed to Section~\ref{sec:proofthm6}.

\begin{theorem}\label{thm:existence}
 Let $\rho \in \mathcal I_N$, $M\in \mathbb{N}^*$ and $\Phi:=(\varphi_1,\ldots,\varphi_M)\in \mathcal F^M$ such that  $\un \in {\rm Span}\{ \Phi \}$. Let us assume in addition that
 \begin{itemize}
     \item[(A$\theta$)] there exists a non-negative non-decreasing continuous function $\theta:\R_+\to\R_+$ such that $\displaystyle \theta(r)\mathop{\longrightarrow}_{r\to +\infty}+\infty$ and $C_\rho:=\int_{\mathbb{R}^3} \theta(|x|) \rho(x)\,dx < +\infty$.
      \end{itemize}

  
For all $C>0$, let us introduce the Moment-Constrained Approximation of the Lieb functional (MCAL)
    \begin{equation}
    \label{min:LiebMomentA}
    \boxed{F_{L,\theta}^{\Phi,C}[\rho]:=\inf_{\substack{\Gamma\in\mathfrak{S}_1^+(\mathcal H_0^N,\Phi,\rho)\\ \tr(\Theta\Gamma)\leq C}}\tr(H_N\Gamma),}
\end{equation}
where $\Theta(x_1,\ldots,x_N):=\frac{1}{N}\sum_{i=1}^N\theta(|x_i|)$ for all $x_1,\ldots,x_N\in \mathbb{R}^3$.
Then, for all $C\geq C_\rho$, $F_{L,\theta}^{\Phi,C}[\rho]$ is finite and a minimum. Moreover, for all $C\geq C_\rho$, there exists a minimizer $\Gamma^{\Phi,C}_{{\rm opt}, \theta}$ to \eqref{min:LiebMomentA} such that $\Gamma^{\Phi,C}_{{\rm opt}, \theta}=\sum_{k=1}^K\omega_k|\Psi_k\rangle\langle \Psi_k|$, for some $1\leq K\leq M+1$, with $\omega_k\geq 0$ and $\Psi_k\in \mathcal H_1^N$ for all $1\leq k\leq K$.
\end{theorem}

 \begin{remark} Let us remark that the existence of a minimizer to a moment-constraint approximation of the Lieb functional has been investigated in~\cite{garrigue2022building}[Theorem~3.1]. More precisely, in the latter work, the author considers moment functions $(\varphi_m)_{m \in \mathcal M}\subset L^\infty(\mathbb{R}^3, \mathbb{R}_+)$, where $\mathcal M$ is a countable subset of $\mathbb{N}^*$, which forms a partition of unity of $\mathbb{R}^3$ i.e. such that
    $$
    \sum_{m\in \mathcal M} \varphi_m = \un. 
    $$
    In particular, $\un \in {\rm Span}\{ \varphi_m, \; m\in \mathcal M\}$. Note that in Theorem~\ref{thm:existence}, assumption (A$\theta$) can be seen as an additional condition on $\rho$ which enables to obtain tightness of minimizing sequences. Instead, the author of~\cite{garrigue2022building} does not require additional conditions on $\rho$ but considers a tightness condition on the set $(\varphi_m)_{m\in \mathcal M}$ which reads as
\begin{equation}\label{eq:tightnessLouis}
    \mathop{\lim}_{R\to +\infty} \sum_{\begin{array}{c}
    m\in \mathcal M\\
    ({\rm Supp} \; \varphi_m) \cap B_R^c \neq \emptyset \\
    \end{array}
    } \int_{\mathbb{R}^3} \rho \varphi_m = 0,
    \end{equation}
    where for all $R>0$, $B_R$ denotes the open ball of $\mathbb{R}^3$ of radius $R$ centered at $0$. 
    Note that our existence result, up to the cost of assuming that $\rho$ satisfies (A$\theta$), allows to treat moment constraints for which the tightness condition (\ref{eq:tightnessLouis}) does not hold. For instance, one can consider a family of moment functions $(\varphi_m)_{1\leq m \leq M}$ where $(\varphi_m)_{1\leq m \leq M-1}$ are the characteristic functions of cells of a mesh associated to a bounded subdomain $\Omega \subset \mathbb{R}^3$ and $\varphi_M = \un_{\Omega^c}$). It can then be easily checked that such a family does not satisfy condition (\ref{eq:tightnessLouis}). 
\end{remark}

\begin{proposition}[Lower semi-continuity]
\label{prop:lsc}
Suppose $\rho_n\in\mathcal I_N$ such that $\rho_{n}\rightharpoonup\rho\in\mathcal I_N$ in $L^1$ then $\liminf F_{L,\theta}^{\Phi,C}[\rho_n]=F_{L,\theta}^{\Phi,C}[\rho]$.
\end{proposition}
\begin{proof}
   The proof is a straightforward adaptation of the proof of Theorem \ref{thm:existence}.
   Assume that $a_n=F_{L,\theta}^{\Phi,C}[\rho_n]\to a$ exists then  up to the extraction of a subsequence, there exists a trace-class operator $\Gamma_\infty \in \mathfrak{S}_1^+(\mathcal H_0^N)$  such that $$\left((H_N+D)^{1/2}\Gamma_n(H_N+D)^{1/2}\right)_{n\in\mathbb{N}}\leq a_n+1/n$$ weakly converges in the sense of trace-class operators to $(H_N+D)^{1/2}\Gamma_\infty(H_N+D)^{1/2}$ as $n$ goes to infinity. Moreover, we have that
   \[\liminf \tr(H_N\Gamma_n)\geq \tr(H_N\Gamma_\infty). \]
   In particular $\Gamma_n$ satisfies the right moment constraints associated to $\rho_n$ as well as $\tr(\Theta\Gamma_n)\leq C$. Then by using the same arguments as in {\bf step 2} of the proof above we deduce that $\Gamma_\infty$ is admissible for $F_{L,\theta}^{\Phi,C}[\rho]$.
   It follows then
   \[F_{L,\theta}^{\Phi,C}[\rho]\leq\tr(H_N\Gamma_\infty)\leq\liminf F_{L,\theta}^{\Phi,C}[\rho_n].\]
\end{proof}

\begin{remark}
    We see from the proof of Theorem~\ref{thm:existence} that assumption (A$\theta$) is needed in order to obtain tightness of the sequence of kernel functions $(\gamma_n)_{n\in \mathbb{N}}$. This is needed because we are considering operators defined on the space $\mathcal H_0^N = \bigwedge_{i=1}^N L^2(\mathbb{R}^3)$. Notice that such a technical assumption is not needed in the case when one considers operators acting on functions acting on a finite domain with Dirichlet boundary conditions. We state such a result below without giving its proof since it follows exactly the same lines as the proof of Theorem~\ref{thm:existence}. 
\end{remark}

Let $\Omega \subset \mathbb{R}^3$ be a bounded subdomain of $\mathbb{R}^3$. We then denote by $\mathcal H_0^N(\Omega):= \bigwedge_{i=1}^N L^2(\Omega)$, $\mathcal H_1^N(\Omega):= \bigwedge_{i=1}^N H^1_0(\Omega)$, $\mathcal H_2^N(\Omega):= \bigwedge_{i=1}^N (H^2(\Omega) \cap H^1_0(\Omega))$ and $\mathcal F(\Omega) := L^\infty(\Omega) + L^{3/2}(\Omega)$. 
The operator $H_{N,\Omega}:= -\frac{1}{2}\Delta + V$ is then a self-adjoint bounded from below operator acting on $\mathcal H_0^N(\Omega)$ with domain $D(H_{N,\Omega}):= \mathcal H_2^N(\Omega)$. 
We also denote by $\mathfrak{S}_1^+(\mathcal H_0^N(\Omega))$ the set of non-negative self-adjoint trace-class operators on $\mathcal H_0^N(\Omega)$. We also define $\mathcal I_N(\Omega)$ the set of function $\rho\in \mathcal I_N$ with support included in $\Omega$. For any $M\in \mathbb{N}^*$ and any $\Phi:=(\varphi_m)_{1\leq m \leq M}\subset \mathcal F(\Omega)$ and $\rho \in \mathcal I_N(\Omega)$, we introduce $\mathfrak{S}_1^+(\mathcal H_0^N(\Omega),\Phi,\rho)$ the set of $\Gamma \in \mathfrak{S}_1^+(\mathcal H_0^N(\Omega))$ such that 
$$
\int_\Omega \rho_\Gamma \varphi_m   = \int_\Omega \rho\varphi_m, \quad \forall 1\leq m \leq M.
$$

Then, the following theorem holds: 

\begin{theorem}\label{thm:existencebounded}
 Let $\rho \in \mathcal I_N(\Omega)$, $M\in \mathbb{N}^*$ and $\Phi:=(\varphi_1,\ldots,\varphi_M)\in \mathcal (\mathcal F(\Omega))^M$ such that  $\un|_\Omega \in {\rm Span}\{ \Phi \}$. Let us introduce
    \begin{equation}
    \label{min:LiebMomentAbounded}
    \boxed{F_{L,\Omega}^{\Phi}[\rho]:=\inf_{\substack{\Gamma\in\mathfrak{S}_1^+(\mathcal H_0^N(\Omega),\Phi,\rho)\\} }\tr(H_{N,\Omega}\Gamma).}
\end{equation}
Then, $F_{L,\Omega}^{\Phi}[\rho]$ is finite and  there exists a minimizer $\Gamma^{\Phi}_{{\rm opt},\Omega}$ to \eqref{min:LiebMomentAbounded} such that $\Gamma^{\Phi}_{{\rm opt},\Omega}=\sum_{k=1}^K\omega_k|\Psi_k\rangle\langle \Psi_k|$, for some $1\leq K\leq M+1$, with $\omega_k> 0$ and $\Psi_k\in \mathcal H_1^N(\Omega)$ for all $1\leq k\leq K$. Moreover, suppose $\rho_n\in\mathcal I_N$ such that $\rho_{n}\rightharpoonup\rho\in\mathcal I_N$ in $L^1$ then $\liminf F_{L,\Omega}^{\Phi}[\rho_n]=F_{L,\Omega}^{\Phi}[\rho]$.

\end{theorem}

In view of the sparsity results we have just proved, it is natural to consider an approximate MCAL problem, where the set of minimizers is restricted to the set of finite-rank operators satisfying moment constraints. More precisely, for a given $K\in \mathbb{N}^*$, we consider the following set
\[
\mathcal O^{C,\Phi_,K}_\theta:=\left\{ \begin{array}{c}
({\bm \omega}, {\bm \Psi})\in \mathbb{R}_+^K \times (\mathcal H_1^N)^K, \quad {\bm \Psi}:=(\Psi_1,\ldots\Psi_K) \in(\mathcal H_1^N)^K,\; \\
    {\bm \omega}:=(\omega_1, \ldots, \omega_K) \in \mathbb{R}_+^K,\\ 
    \widetilde{\rho}:=  \sum_{k=1}^K \omega_k \rho_{\Psi_k}, \quad \int_{\mathbb{R}^3} \widetilde{\rho}(x) \theta(|x|)\,dx \leq C,\\
    \forall 1\leq m \leq M, \;    \int_{\mathbb{R}^3} \varphi_m \widetilde{\rho} = \int_{\mathbb{R}^3} \varphi_m \rho\\
\end{array}
\right\}.
\]

The approximate MCAL functional then reads as follows

 \begin{equation}
    \label{min:LiebMomentbis}
    \boxed{F_{L,\theta}^{\Phi,C,K}[\rho]:=\inf_{({\bm \Psi}, {\bm \omega})\in\mathcal O^{C,\Phi_,K}_\theta
    }
    \mathcal J({\bm \Psi}, {\bm \omega}),}
\end{equation}
where 
$$
\mathcal J({\bm \Psi}, {\bm \omega}):= \sum_{k=1}^K \omega_k \langle \Psi_k |H_N|\Psi_k\rangle.
$$
\begin{remark}
    Notice that as soon as $K\geq M+1$ then we have that $F_{L,\theta}^{\Phi,C,K}[\rho]=F_{L,\theta}^{\Phi,C}[\rho]$.
\end{remark}
\begin{remark}
    Since $\rho\in\mathcal I_N$ then  the set $\mathcal O^{C,\Phi_,K}_\theta$ is no empty. Moreover it can be shown, by standard arguments, that there exists a minimizer to \eqref{min:LiebMomentbis}.
\end{remark}
As in the case of moment constrained optimal transport \cite{alfonsi2022constrained} we can state some interesting mathematical properties on the set of minimizers of the approximate problem \eqref{min:LiebMomentbis}. First, consider two elements of $\mathcal O^{C,\Phi_,K}_\theta$, then there exists a continuous path in $\mathcal O^{C,\Phi_,K}_\theta$ connecting these two elements  and such that $\mathcal J$ varies monotonically along it.
\begin{theorem}
\label{thm:monotone}
    Let us assume that $K\geq 2M+2$. Let $({\bm \Psi}_0, {\bm \omega}_0),({\bm \Psi}_1, {\bm \omega}_1)\in \mathcal O^{C,\Phi_,K}_\theta$. Then, there exists a continuous application $\eta:[0,1]\to \mathcal O^{C,\Phi_,K}_\theta$ made of polygonal chain such that $\eta(0)=({\bm \Psi}_0, {\bm \omega}_0)$, $\eta(1)=({\bm \Psi}_1, {\bm \omega}_1)$
    and such that the application $t\mapsto \mathcal J(\eta(t))$ is monotone.
\end{theorem}

\noindent Since the proof is a straightforward adaptation of the one for \cite{alfonsi2022constrained}[Theorem 1], we refer the reader to it. We only highlight that, as we did in the previous sections, given a couple  $({\bm \Psi}, {\bm \omega})$ one can always associate a measure $\mu=\sum_{i}^K\omega_i\delta_{\psi_i}$, then by Thchakaloff's theorem the result follows.
An interesting consequence of theorem \ref{thm:monotone} concerns the  minimizers of MCAL: first, as soon as $K\geq 2M+2$ any local minimizer of MCAL (or of problem \eqref{min:LiebMomentbis}) is a global minimizer. Secondly, the set of minimizers forms a polygonally connected set.
\begin{corollary}
Assume that $K\geq 2M+2$. Then, any local minimizer of \eqref{min:LiebMomentbis} is a global minimizer. Moreover, the set of minimizers of \eqref{min:LiebMomentbis} is a polygonally connected subset of $\mathcal O^{C,\Phi_,K}_\theta$.
\end{corollary}

\section{Some convergence results}\label{sec:convergence}

The aim of this section is to gather some convergence results on the MCAL approximation towards solutions of the exact problem.

\subsection{Convergence of the MCAL functional to the exact Lieb functional}\label{sec:convergenceLieb}
The aim of this section is to prove that, under some appropriate assumptions, the MCAL functional converges to the exact Lieb functional as the number of moment constraints go to infinity. Let us denote here by $\mathcal D(\mathbb{R}^3)$ the set of $\mathcal C^\infty$ real-valued functions defined on $\mathbb{R}^3$ with compact support.

\medskip

More precisely, let $\rho\in \mathcal I_N$ such that there exists a function $\theta : \mathbb{R}_+ \to \mathbb{R}_+$ satisfying assumption (A$\theta$). Let $C_\rho:= \int_{\mathbb{R}^3} \theta(|x|) \rho(x)\,dx$ and let $C>C_\rho$.

\medskip

For all $n\in \mathbb{N}^*$, let $M_n \in \mathbb{N}^*$ and $\Phi^n:=(\varphi^n_m)_{1\leq m\leq M_n} \subset \mathcal F$ be a sequence of functions belonging to $\mathcal F$ and which satisfies $\un \in {\rm Span}\{\Phi^n\}$ for all $n\in \mathbb{N}^*$ together with the following density conditions: 
\begin{itemize}
    \item [(A$\Phi$)] for all $f\in \mathcal D(\mathbb{R}^3)$, 
    $$
    \mathop{\inf}_{g_n \in {\rm Span} \{\Phi^n\}} \|f -g_n\|_{\mathcal F} \mathop{\longrightarrow}_{n\to +\infty} 0 . 
    $$
\end{itemize}

Then, we have the following useful lemma that we will use in the sequel. 

\begin{lemma}\label{lem:useful}
    Let $(\widetilde{\rho}_n)_{n\in \mathbb{N}^*} \subset \mathcal I_N$ such that $\mathop{\sup}_{n\in \mathbb{N}^*} \|\sqrt{\widetilde{\rho}_n}\|_{H^1(\mathbb{R}^3)} < +\infty$ and such that for all $n\in \mathbb{N}^*$, 
    $$
   \forall g_n \in {\rm Span}\{\Phi^n\}, \quad \int_{\mathbb{R}^3} \widetilde{\rho}_n g_n = \int_{\mathbb{R}^3} \rho g_n.
    $$
Then, $(\widetilde{\rho}_n)_{n\in \mathbb{N}^*}$ converges in the sense of distributions to $\rho$ as $n$ goes to infinity. 
\end{lemma}

\begin{proof}
The proof uses the same lines as the proof of~\cite{garrigue2022building}[Theorem~3.2]. We rewrite it here for the sake of completeness. Let $f\in \mathcal D(\mathbb{R}^3)$ and let $(f_n)_{n\in\mathbb{N}^*}$ be a sequence of functions such that $f_n\in {\rm Span}\{\Phi^n\}$ for all $n\in \mathbb{N}^*$ and $\displaystyle \|f-f_n\|_{\mathcal F} \mathop{\longrightarrow}_{n\to +\infty} 0$. Then, it holds that
\begin{align*}
    \left| \int_{\mathbb{R}^3} f (\widetilde{\rho}_n - \rho) \right| &=   \left| \int_{\mathbb{R}^3} (f-f_n) (\widetilde{\rho}_n - \rho) \right| \\
    & \leq C \left( \|\sqrt{\rho}\|_{H^1(\mathbb{R}^3)}^2 + \mathop{\sup}_{n\in \mathbb{N}^*} \|\sqrt{\widetilde{\rho}_n}\|_{H^1(\mathbb{R}^3)}^2\right) \|f-f_n\|_{\mathcal F},\\
    & \mathop{\longrightarrow}_{n\to +\infty} 0.\\
\end{align*}
Hence the desired result.
\end{proof}




\begin{remark}
One example of sequence $(\Phi_n)_{n\in \mathbb{N}^*}$ satisfying (A$\Phi$) is the following: for all $n\in \mathbb{N}^*$, let $\Omega_n:= (-n,n)^3$ and let $\mathcal T_n:=\{ T_1^n, \ldots, T_{N_n}\}$ (with $N_n:=\# \mathcal T_n$) be a regular conforming triangular mesh of $\Omega_n$, the elements of which have a maximal diameter size $h_n$ such that $h_n \leq \frac{1}{n}$. Let $M_n:= \# \mathcal T_n +1 = N_n+1$. Denoting by 
$\varphi_m^n:= \un|_{T_m^n}$ for $1\leq m \leq M_n-1$ and by $\varphi_{M_n}^n:= \un|_{\Omega_n^c}$ and by $\Phi^n = (\varphi_m^n)_{1\leq m \leq M_n}$ for all $n\in \mathbb{N}^*$, one can easily check that the sequence $(\Phi^n)_{n\in \mathbb{N}^*}$ satisfies (A$\Phi$).
\end{remark}

We then have the following convergence result, which may be seen as an extension of~\cite{garrigue2022building}[Theorem~3.2] to more general set of moment functions, up to the additional tightness assumption (A$\theta$), the proof of which is postponed to Section~\ref{sec:convergencethm}.

\begin{theorem}
\label{thm:convergence}
    Let $\rho \in \mathcal I_N$ such that there exists a function $\theta: \mathbb{R}_+ \to \mathbb{R}_+$ satisfying assumption (A$\theta$). Let $C_\rho:= \int_{\mathbb{R}^3} \theta(|x|)\rho(x)\,dx$ and $C\geq C_\rho$. For all $n\in \mathbb{N}^*$, let $M_n \in \mathbb{N}^*$ and $\Phi^n:= (\varphi_m^n)_{1\leq m \leq M_n} \subset \mathcal F$ such that assumption (A$\Phi$) holds. We assume in addition that there exists $n_0\in \mathbb{N}^*$ such that $\un \in {\rm Span}\{\Phi^n\}$ for all $n\geq n_0$. 
    Then, for all $n\geq n_0$, there exists at least one sparse minimizer to (\ref{min:LiebMomentA}) with $\Phi = \Phi^n$ in the sense of Theorem~\ref{thm:existence}. Besides, it holds that
  \begin{equation}
      \mathop{\lim}_{n\to+\infty}F_{L,\theta}^{\Phi^n,C}[\rho]=F_{L}[\rho].
  \end{equation}
    Moreover, from any sequence $(\Gamma_n)_{n\geq n_0}$ such that $\Gamma_n$ is a minimizer for \eqref{min:LiebMomentA} with $\Phi = \Phi^n$, one can extract a subsequence which strongly converges in $\mathfrak{S}_{1,1}(\mathcal H_0^N)$ to $\Gamma_\infty$, where $\Gamma_\infty$ is a minimizer of \eqref{min:lieb}.
\end{theorem}

Like in Section~\ref{sec:sparsity}, we can state a similar result with less technical assumptions in the case when we consider operators acting on functions defined on a bounded subdomain $\Omega \subset \mathbb{R}^3$ with Dirichlet boundary conditions. We state such a result here, using the same notation as in Section~\ref{sec:sparsity}, since it follows exactly the same lines of proof as Theorem~\ref{thm:convergence}. 
To this aim, for all $\rho \in \mathcal I_N(\Omega)$, we introduce the exact Lieb functional defined on the domain $\Omega$ as
\begin{equation}
    \label{min:liebbounded}
    F_{L,\Omega}[\rho]:=\inf_{\substack{\Gamma \in \mathfrak{S}_1^+(\mathcal H_0^N(\Omega))\\ \rho_{\Gamma}=\rho}}\tr(H_{N,\Omega}\Gamma).
\end{equation}
Let us point out here that there exists also $\epsilon_\Omega, D_\Omega >0$ such that 
$$
H_{N,\Omega} + D_\Omega \geq \varepsilon_\Omega(-\Delta_\Omega + 1)
$$
where $-\Delta_\Omega$ refers here to the self-adjoint bounded from below operator on $\mathcal H_0^N(\Omega)$ with domain $\mathcal H_N^2(\Omega)$ (Laplacian with Dirichlet boundary conditions in $\Omega$). We also denote by $\mathfrak{S}_{1,1}\left( \mathcal H_0^N(\Omega)\right)$ the set of operators $\Gamma \in \mathfrak{S}_1^+(\mathcal H_0^N(\Omega))$ such that ${\rm Tr}(-\Delta_\Omega \Gamma) < +\infty$.

\begin{theorem}
\label{thm:convergencebounded}
      Let $\rho \in \mathcal I_N(\Omega)$. For all $n\in \mathbb{N}^*$, let $M_n \in \mathbb{N}^*$ and $\Phi^n:= (\varphi_m^n)_{1\leq m \leq M_n} \subset \mathcal F(\Omega)$ such that for all $f\in \mathcal D(\Omega)$, 
      $$
     \lim_{n\to+\infty} \mathop{\inf}_{g_n\in {\rm Span}\{\Phi^n\}} \|f - g_n\|_{\mathcal F(\Omega)}  = 0. 
     $$
      We assume in addition that there exists $n_0\in \mathbb{N}^*$ such that $\un \in {\rm Span}\{\Phi^n\}$ for all $n\geq n_0$. 
    Then, for all $n\geq n_0$, there exists at least one sparse minimizer to (\ref{min:LiebMomentA}) with $\Phi = \Phi^n$ in the sense of Theorem~\ref{thm:existence}. Besides, it holds that
  \begin{equation}
      \mathop{\lim}_{n\to+\infty}F_{L,\Omega}^{\Phi^n}[\rho]=F_{L,\Omega}[\rho].
  \end{equation}
    Moreover, from any sequence $(\Gamma_n)_{n\geq n_0}$ such that $\Gamma_n$ is a minimizer for \eqref{min:LiebMomentA} with $\Phi = \Phi^n$, one can extract a subsequence which strongly converges in $\mathfrak{S}_{1,1}(\mathcal H_0^N(\Omega))$ to $\Gamma_\infty$, where $\Gamma_\infty$ is a minimizer of \eqref{min:lieb}.
\end{theorem}

 \subsection{Convergence rate of the ground state energy in the bounded domain case}\label{sec:convergenceGround}

In this section, we restrict ourselves to the case of a bounded subdomain $\Omega \subset \mathbb{R}^3$. Let $M\in \mathbb{N}^*$, $\Phi:=(\varphi_m)_{1\leq m \leq M} \subset \mathcal F(\Omega)$ be a set of moment functions. For all $v\in \mathcal F(\Omega)$, let us introduce the ground state energy associated to the potential $v$: 
\[E[v]:=\inf_{\Psi\in\mathcal H_1^N(\Omega)}\langle\Psi | H_{N,\Omega}^v|\Psi \rangle=\inf_{\Gamma\in \mathfrak{S}_1^+(\mathcal H_0^N(\Omega))}\tr(H_{N,\Omega}^v\Gamma),\]
where $$
H_{N,\Omega}^v:= H_{N,\Omega} -  \sum_{i=1}^N v(x_i).
$$
Rewriting the minimization over $\Gamma$ as an external minimization over $\rho\in \mathcal I_N(\Omega)$ and then as an internal one over all $\Gamma$ such that ${\rm Tr}\; \Gamma = \rho$, it can easily be checked that
\begin{equation}\label{eq:defE}
E[v] = \mathop{\inf}_{\rho \in \mathcal I_N(\Omega)} \left\{ F_L[\rho] - \int_{\Omega} v\,d\rho \right\}.
\end{equation}
Let us also define by
\begin{equation}
E^\Phi[v]:= \mathop{\inf}_{\rho \in \mathcal I_N(\Omega)} \left\{ F^\Phi_L[\rho] - \int_{\Omega} v\,d\rho \right\}.
\end{equation}

Similarly, let us point out that, if $v\in {\rm Span}\{\Phi\}$, rewriting the minimization over $\Gamma$ as an external minimization over $\rho\in \mathcal I_N(\Omega)$ and then as an internal one over all $\Gamma\in \mathfrak{S}_1^+(\mathcal H_0^N(\Omega), \Phi, \rho)$, it holds that
$$
E[v]= E^\Phi[v], \quad \forall v \in {\rm Span}\{\Phi\}.
$$

We then prove the following approximation result. 

\begin{proposition}\label{prop:error}
Let us assume that $v\in L^\infty(\Omega)$ and that $\Phi = (\varphi_m)_{1\leq m \leq M} \subset L^\infty(\Omega)$. Then, it holds that
\begin{equation}
|E[v] - E^\Phi[v]| \leq 2 N \mathop{\min}_{w\in {\rm Span}\{\Phi\}} \|v - w\|_{L^\infty(\Omega)}.
\end{equation}
\end{proposition}

\begin{proof}
Let $\displaystyle v^\Phi = \mathop{\rm argmin}_{w\in {\rm Span}\{\Phi\}} \|v -w\|_{L^\infty(\Omega)}$. Let $\varepsilon>0$ arbitrarily small. Let $\rho$, $\rho^\Phi$, $\widetilde{\rho}^\Phi$ and  $\overline{\rho}^\Phi$ be $\varepsilon$-minimizers of $E[v]$, $E[v^\Phi]$, $E^\Phi[v]$ and $E^\Phi[v^\Phi]$ respectively. It then holds that
\begin{align*}
E[v^\Phi] & \leq F_L[\rho^\Phi] - \int_{\Omega} v^\Phi\,d\rho^\Phi \\
& \leq E[v_\Phi] + \varepsilon \\
& \leq F_L[\rho] - \int_{\Omega} v^\Phi\,d\rho +\varepsilon \\
& = F_L[\rho] - \int_{\Omega} v\,d\rho + \int_\Omega (v^\Phi -v)\,d\rho + \varepsilon \\
& \leq E[v] + \int_\Omega (v^\Phi -v)\,d\rho + 2\varepsilon.
\end{align*}
Using similar calculations, we obtain that 
$$
E[v] \leq E[v^\Phi] + \int_\Omega (v- v^\Phi)\,d\rho^\Phi + 2 \varepsilon.
$$
As a consequence, we obtain that
$$
|E[v] - E[v^\Phi]| \leq \max\left( \int_\Omega |v-v^\Phi|\,d\rho, \int_\Omega |v-v^\Phi|\,d\rho^\Phi \right) + 2 \varepsilon \leq N \|v-v^\Phi\|_{L^\infty(\Omega)} + 2\varepsilon.
$$
Since $\varepsilon$ can be chosen arbitrarily small, it actually holds that 
\begin{equation}\label{eq:first}
|E[v] - E[v^\Phi]| \leq  N \|v-v^\Phi\|_{L^\infty(\Omega)}.
\end{equation}
Using similar arguments, we also obtain that 
\begin{equation}\label{eq:second}
|E^\Phi[v] - E^\Phi[v^\Phi]| \leq N \|v-v^\Phi\|_{L^\infty(\Omega)}.
\end{equation}
Collecting (\ref{eq:first}) and (\ref{eq:second}) and using the fact that $E[v^\Phi]= E^\Phi[v^\Phi]$ yields the desired result.
\end{proof}

Proposition~\ref{prop:error} then enables to quantify the rate of convergence of $|E[v] - E^{\Phi^n}[v]|$ as $n$ goes to infinity for some particular sequences of moment functions $(\Phi^n)_{n\in\mathbb{N}}$ provided that $v$ is regular enough. As an illustration, we analyze here the rate of convergence of a numerical method inspired from the external dual charge approach recently proposed in~\cite{lelotte2022external}. 

\begin{corollary}\label{cor:cor1}
Let $l\geq 0$ and $\Omega$ be a bounded regular subdomain of $\mathbb{R}^3$. Let $\mu \in H^{l+1}(\Omega)$ be an external density of charge and define $v\in H^1_0(\Omega) \cap H^{l+3}(\Omega)$ as the unique solution to 
$$
\left\{
\begin{array}{ll}
-\Delta v = \mu &  \quad \mbox{in }\Omega,\\
v=0 & \quad \mbox{on }\partial \Omega.\\
\end{array}
\right.
$$
 Let $(\mathcal T_h)_{h>0}$ be a sequence of triangular regular meshes of $\Omega$ such that
 $$
 h:= \max_{K \in \mathcal T_h} {\rm diam}(K).
 $$
Let $k\in \mathbb{N}$ and $P_h^k\subset L^\infty(\Omega)$ be the subspace of continuous $\mathbb{P}_k$ finite element functions associated to the mesh $\mathcal T_h$. We denote by $V_{h,k}$ the subspace of $H^1_0(\Omega) \cap H^2(\Omega)$ containing all functions $v_{h,k} \in H^1_0(\Omega) \cap H^2(\Omega)$ solution to
$$
\left\{
\begin{array}{ll}
-\Delta v_{h,k} = \mu_{h,k} &  \quad \mbox{in }\Omega,\\
v_{h,k}=0 & \quad \mbox{on }\partial \Omega,\\
\end{array}
\right.
$$
for some $\mu_{h,k}\in P_h^k$. Let $\Phi_{h,k}$ be a basis of $V_{h,k}$. Then, asuming that $l\leq k$, there exists a constant $C>0$ such that for all $h>0$, 
$$
\boxed{|E[v]- E^{\Phi_{h,k}}[v]| \leq C N h^{l+1} \|v\|_{H^{l+3}(\Omega)}.}
$$
\end{corollary}

\begin{proof}
Corollary~\ref{cor:cor1} easily follows for the compact embedding $H^2(\Omega) \hookrightarrow L^\infty(\Omega)$ and standard interpolation error results associated with finite element approximations.
\end{proof}

\begin{remark}
Denoting by $M_{h,k}$ the dimension of $V_{h,k}$, it holds that $M_{h,k} = \mathcal O\left( \frac{k}{h^3} \right)$. As a consequence, the above result implies that the rate of convergence of $E^{\Phi_{h,k}}[v] $ 
to $E[v]$ decays like $\mathcal O\left( \frac{N}{M_{h,k}^{(l+1)/3}}\right)$ where $M_{h,k}$ is the number of moment constraints in the MCAL approximation.

\end{remark}


%
\section{Duality results for the MCAL functional}\label{sec:duality}
%
Let us begin by recalling some classical results about semi-definite programming problems and introduce some notation. 

\subsection{Semi-definite positive programming problems}\label{sec:semidef}

Let $n\in \mathbb{N}^*$. We denote by $\mathcal S^n$ the set of symmetric matrices of $\mathbb{R}^n$. For any $M\in \mathcal S^n$, the notation $M \succcurlyeq 0$ (respectively $M  \succ 0$) is used to mean that $M$ is a semi-definite non-negative (respectively definite positive) matrix. We also denote by $\mathcal S^n_{+}:=\{ M\in \mathcal S^n, \; M \succcurlyeq 0\}$ and by $\mathcal S^n_{+,*}:= \{ M\in \mathcal S^n, \; M \succ 0\}$. For all $M,N\in \mathcal S^n$, we denote by $\langle M, N \rangle = {\rm Tr}(M^T N)$ the Frobenius scalar product between $M$ and $N$. 

\medskip

Let $m \in \mathbb{N}^*$, $C\in \mathcal S^n$, $A: \mathcal S^n \to \mathbb{R}^m$ a linear application and $b\in \mathbb{R}^m$. We consider here the following (primal) semi-definite positive programming problem: 
\begin{equation}\label{eq:primalSDP}
\boxed{P:=\mathop{\inf}_{\begin{array}{c}
X\in \mathcal S^n \\
A(X) = b\\
X \succcurlyeq 0\\
\end{array}
} \langle C, X \rangle. }
\end{equation}
The dual problem associated to (\ref{eq:primalSDP}) then reads as follows:
\begin{equation}\label{eq:dualSDP}
\boxed{D:=\mathop{\sup}_{
\begin{array}{c}
(y ,S) \in \mathbb{R}^m \times \mathcal S^n\\
A^*(y) + S = C \\
S\succcurlyeq 0\\
\end{array}
}
\langle b , y \rangle}
\end{equation}
where $A^*: \mathbb{R}^m \to \mathcal S^n$ is the adjoint of $A$.

We introduce the following sets: 
\begin{align*}
\mathcal A_P& :=\left\{X \in \mathcal S^n, \; A(X) = b, \; X \succcurlyeq 0\right\},\\
\mathcal A^s_P& :=\left\{X \in \mathcal S^n, \; A(X) = b, \; X \succ 0 \right\},\\
\mathcal A_D& :=\left\{(y ,S) \in \mathbb{R}^m \times \mathcal S^n, \;  A^*(y) + S = C, \; S\succcurlyeq 0 \right\},\\
\mathcal A^s_D& :=\left\{ (y ,S) \in \mathbb{R}^m \times \mathcal S^n, \;  A^*(y) + S = C, \; S\succ 0 \right\}.\\
\end{align*}

We also denote by ${\rm Sol}_P$ and ${\rm Sol}_D$ the set of solutions to (\ref{eq:primalSDP}) and (\ref{eq:dualSDP}). Then, we recall the following classical result~\cite{anjos2011handbook,wolkowicz2012handbook}:  

\begin{theorem}\label{th:SDPex}
    \begin{itemize}
\item[(i)] If $\mathcal A_P \times \mathcal A^s_D \neq \emptyset$, ${\rm Sol}_P$ is non-empty and bounded and $P =D$;
\item[(ii)] If $\mathcal A^s_P \times \mathcal A_D \neq \emptyset$ and $A$ surjective, then ${\rm Sol}_D$ is non-empty and bounded and $P =D$;
\item[(iii)] If $\mathcal A^s_P \times \mathcal A^s_D \neq \emptyset$ and $A$ surjective, then ${\rm Sol}_P$ and ${\rm Sol}_D$ are non-empty and bounded and $P =D$.
    \end{itemize}
\end{theorem}

\subsection{Dual MCAL problem}

In this section we study the dual
problem in the bounded domain case. We know that the dual variable associated to the density $\rho\in\mathcal I_N(\Omega)$ is a one-body interaction potential of the form 
$W^v(x_1,\ldots,x_N):=\sum_{i=1}^Nv(x_i)$ for a given $v\in \mathcal F(\Omega)$. 

We then consider the following natural dual problem
\begin{equation}
\label{eq:dual}
\boxed{D_{L,\Omega}^{\Phi}[\rho]=\sup_{
\begin{array}{c}
v \in {\rm Span}\{\Phi\},\\
\forall \Psi\in \mathcal H_1^N(\Omega), \;\langle\Psi | H_{N,\Omega}^v|\Psi \rangle \geq 0 \\
\end{array}}\int_\Omega vd\rho,}
\end{equation}
where $$
H_{N,\Omega}^v:= H_{N,\Omega} -  \sum_{i=1}^N v(x_i) = H_{N,\Omega} - W^v.
$$
If we take any $v:= \sum_{m=1}^M \alpha_m \varphi_m \in {\rm Span}\{\Phi\}$ satisfying the above constraints and any $\Gamma\in \mathfrak{S}_1^+(\mathcal H_0^N(\Omega),\Phi,\rho)$  then we have
\[
\begin{split}
 \tr(H_{N,\Omega}\Gamma)\geq \tr(W^v\Gamma)&=\int_\Omega vd\rho_{\Gamma}=\int_\Omega\bigg(\sum_{m=1}^M\alpha_m\varphi_m\bigg)d\rho_\Gamma\\
 &\geq\int_\Omega\bigg(\sum_{m=1}^M\alpha_m\varphi_m)\bigg)d\rho=\int_\Omega vd\rho   
\end{split}
 \]
which proves that $F_{L,\Omega}^{\Phi}[\rho]\geq D_{L,\Omega}^{\Phi,C}[\rho]$.
We would like to prove that this inequality is actually an equality.
Let us introduce the ground state energy associated to the potential $v$: 
\[E[v]=\inf_{\Psi\in\mathcal H_1^N(\Omega)}\langle\Psi | H_{N}^v|\Psi \rangle=\inf_{\Gamma\in \mathfrak{S}_1^+(\mathcal H_0^N(\Omega))}\tr(H_N^v\Gamma).\]
We rewrite now the minimization over $\Gamma$ as an external minimization over $\rho\in \mathcal I_N(\Omega)$ and then as an internal one over all $\Gamma$ in $\mathfrak{S}_1^+(\mathcal H_0^N(\Omega),\Phi,\rho)$ (we are considering the ground state for a potential $v\in {\rm Span}\{\Phi\}$):
\[E[v]=\inf_{\rho\in\mathcal I_N(\Omega)}\left\{F_{L,\Omega}^{\Phi}[\rho]-\int_\Omega v\dd\rho\right\}.\]
Notice that $E$ is nothing but the Legendre-Fenchel transform of 
$F_{L,\Omega}^{\Phi}[\rho]$.
On the other hand, we rewrite \eqref{eq:dual}  in the form
\begin{equation}
    \label{eq:dualbis}
    D_{L,\Omega}^{\Phi}[\rho]=\sup_{v\in{\rm Span}\{\Phi\}}\left\{\int_\Omega v\dd\rho-E[v]\right\}.
\end{equation}
Thus, $D_{L,\Omega}^{\Phi}[\rho]$ is the Legendre transform of $E$. From Proposition~\ref{prop:lsc} and Fenchel duality theorem for convex lower semi-continuous functions we conclude the following
\begin{theorem}
    \label{thm:stronDuality}
    Under the assumptions of Theorem~\ref{thm:existencebounded}, we have $F_{L,\Omega}^{\Phi}[\rho]=D_{L,\Omega}^{\Phi}[\rho]$.
\end{theorem}

We now have the following result which, taking into account the sparsity result of Theorem~\ref{thm:existencebounded}, gives a more convenient formulation of $D_{L,\Omega}^{\Phi}[\rho]$. 

\begin{theorem}\label{thm:dual}
Under the assumptions of Theorem~\ref{thm:existencebounded}, there exists at least one maximizer to (\ref{eq:dual}), and it holds that
\begin{align*}
D_{L,\Omega}^{\Phi}[\rho] & = \mathop{\max}_{
\begin{array}{c}
v \in {\rm Span}\{\Phi\},\\
\forall \Psi\in \mathcal H_1^N(\Omega), \quad \langle\Psi | H_{N,\Omega}^v|\Psi \rangle \geq 0 \\
\end{array}} \int_\Omega v \rho \\
& = \mathop{\max}_{
\begin{array}{c}
v \in {\rm Span}\{\Phi\},\\
\forall \Psi \in {\rm Span}\{\Psi_1, \ldots, \Psi_K\}, \quad \langle\Psi | H_{N,\Omega}^v|\Psi \rangle \geq 0 \\
\end{array}} \int_\Omega v \rho, \\
\end{align*}
where 
\[\Gamma^{\Phi}_{{\rm opt},\Omega}=\sum_{k=1}^K\omega_k|\Psi_k\rangle\langle \Psi_k|\]
for some $1\leq K\leq M+1$, with $\omega_k> 0$ and $\Psi_k\in \mathcal H_1^N(\Omega)$ for all $1\leq k\leq K$ is a minimizer of (\ref{min:LiebMomentAbounded}).
\end{theorem}




 \section{Numerical scheme}\label{sec:algo}
The aim of this section is to propose a new numerical scheme using the sparsity of minimizers of the MCAL functional to compute approximations of the Lieb functional. 
The scheme proposed here requires the resolution of eigenvalue problems for operators acting on $\mathcal H_0^N(\Omega)$, which leads to high-dimensional problems when the number of electrons is large. 
The combination of the algorithm proposed here with numerical methods dedicated to overcome the curse of dimensionality will be the object of a future work.

\medskip

We propose here an iterative scheme which shares some common features with the well-known Column Algorithm used for classical optimal transport problems (see for instance~\cite{friesecke2022gencol,friesecke2022genetic}). 
The aim is to construct at each iteration $n\in \mathbb{N}^*$ a finite set of $L^2$-normalized wavefunctions $\mathfrak{P}_n \subset \mathcal H^2_N(\Omega)$ which will be used to enforce the inequality constraints in the resolution of the MCAL dual problems. More precisely, inequality constraints in small-dimensional dual problems are enforced to hold on the space spanned by the wavefunctions belonging to the set $\mathfrak{P}_n$. As a consequence, in our present quantum optimal transprt framework, semi-definite programming problems have to be solved at each iteration instead of linear programming problems for classical optimal transport problems.

\subsection{MCAL iterative scheme}

We describe the MCAL algorithm in this section. The algorithm takes as input data:  
\begin{itemize}
    \item $\Phi = (\varphi_1, \ldots, \varphi_M)\subset L^\infty(\Omega)$ set of moment functions; 
    \item $\rho \in \mathcal I_N$ with support included in $\Omega$; \item $\rho^m:= \int_\Omega \varphi_m \rho$ for all $1\leq m \leq M$.  
      \item $\widetilde{\mathfrak{P}}_0 \subset \mathcal H^2_N(\Omega)$ initial finite set of $L^2$-normalized wavefunctions.
          \end{itemize}

          \medskip

As an output, after $n$ iterations, the algorithm yields $F^n$ which is an approximation of the quantity $F_{L,\Omega}^\Phi[\rho]$. 

\medskip

     We make here the following assumption on the initial set $\widetilde{\mathfrak{P}}_0$.

    {\bf Assumption (A0):} let $\widetilde{K}_0:= {\rm dim} \; {\rm Span}\{\widetilde{\mathfrak{P}}_0\}$ and $(\widetilde{\Psi}_1^0, \ldots, \widetilde{\Psi}_{K_0}^0)$ an orthonormal basis of ${\rm Span}\{\widetilde{\mathfrak{P}}_0\}$. We assume that there exists $\widetilde{S}:= (\widetilde{S}_{kl})_{1\leq k,l\leq K_0}\in \mathcal S_+^{K_0}$ such that for all $1\leq m \leq M$, 
      $$
      \sum_{k,l=1}^{K_0} \widetilde{S}_{kl} \int_{\Omega} \varphi_m \overline{\widetilde{\Psi}_k^0}\widetilde{\Psi}_l^0 = \rho^m.
      $$

In the case when $d = 3$, a way to find such an initial set $\widetilde{\mathfrak{P}}_0$ is given in~\cite{Lieb-83}. In this case, $\widetilde{K}_0$ can be chosen to be equal to $1$ and $\widetilde{\Psi}_1^0$ can be chosen as follows: for all $1\leq k \leq N$ and $x = (x_1,x_2,x_3)\in \mathbb{R}^3$, define
$$
\phi^k(x) = \sqrt{\frac{\rho(x)}{N}}e^{ikf(x_1)},
$$
where for all $x_1\in \mathbb{R}$, 
$$
f(x_1) = \left(\frac{2\pi}{N}\right) \int_{-\infty}^{x_1} \,ds \int_{-\infty}^{+\infty}\,dt \int_{-\infty}^\infty \,du \rho(s,t,u).
$$
The family $(\phi^k)_{1\leq k \leq N}$ then forms an orthonormal family of $L^2(\mathbb{R}^3)$ and one may define $\widetilde{\Psi}_1^0$ as the normalized Slater determinant associated to the family $(\phi^k)_{1\leq k \leq N}$.

\bigskip

\subsubsection{Initialization step} 

   Compute $\widetilde{S}^0:=(\widetilde{S}^0_{kl})_{1\leq k ,l \leq K_0}\in \mathcal S_+^{K_0}$ solution to 
\begin{equation}\label{eq:primal}
    \widetilde{F}^0:=\mathop{\min}_{
   \begin{array}{c}
   (S_{kl})_{1\leq k,l \leq K_n}\in \mathcal{S}_+^{K_0},\\
   \forall 1\leq m \leq M, \\
\sum_{k,l=1}^{K_0} S_{kl} \int_{\Omega} \varphi_m \overline{\widetilde{\Psi}_k^0}\widetilde{\Psi}_l^0=  \rho^m\\
   \end{array}
   } \sum_{k,l=1}^{K_0} S_{kl} \langle \widetilde{\Psi}_k^0 | H_{N,\Omega} | \widetilde{\Psi}_l^0\rangle
   \end{equation}
Then, it holds that $\widetilde{S}^0 = \sum_{k=1}^{K_0} \omega_k^0 (U_k^0) (U_k^0)^T$
where $(\omega_k^0)_{1\leq k \leq K_0}\in \mathbb{R_+}^{K_0}$ are the 
eigenvalues of $\widetilde{S}^0$ (assumed to be ranked in non-increasing
order) and for all $1\leq k \leq K_0$, $U_k^0:= (U_{kl}^0)_{1\leq l \leq 
K_0}\in \mathbb{R}^{K_0}$ is a normalized eigenvector associated with 
$\omega_k^0$ so that $(U_1^0, \ldots, U_{K_0}^0)$ forms an orthonormal 
basis of $\mathbb{R}^{K_0}$. 

\medskip

Let $\mathcal K_0:= \max \left\{ k\in \{1, \ldots, K_0\}, \omega_k^0 > 0\right\}$. For all $1\leq k \leq \mathcal K_0$, let 
$
\Psi_k^0:= \sum_{l=1}^{K_0} U_{kl}^0 \widetilde{\Psi}_l^0,
$
and $S^0:= {\rm diag}(\omega_1^0, \ldots, \omega_{\mathcal K_0})\in \mathcal S_{+,*}^{\mathcal K_0}$. We also denote by $\mathfrak{P}_0:= \bigcup_{1\le k \leq \mathcal K_0} \left\{\Psi_k^0\right\}$.

\begin{remark}
    It is easy to see that, by construction, it holds that
\begin{equation}\label{eq:primal2}
    \widetilde{F}^0:=\mathop{\min}_{
   \begin{array}{c}
   (S_{kl})_{1\leq k,l \leq \mathcal K_0}\in \mathcal{S}_+^{\mathcal K_0},\\
   \forall 1\leq m \leq M, \\
\sum_{k,l=1}^{\mathcal K_0} S_{kl} \int_{\Omega} \varphi_m \overline{\Psi_k^0}\Psi_l^0=  \rho^m\\
   \end{array}
   } \sum_{k,l=1}^{K_0} S_{kl} \langle \Psi_k^0 | H_{N,\Omega} | \Psi_l^0\rangle
   \end{equation}
and that $S^0$ is a minimizer to (\ref{eq:primal2}). 
\end{remark}

\begin{remark}
   Notice also that this initialization step is useless in the case when $\widetilde{K}^0 = 1$.
\end{remark}
   
\subsubsection{Iteration $n\geq 1$}

{\itshape \bfseries Step~1:} Let $\mathcal K^{n-1}:={\rm dim}\; {\rm Span}\left\{ \mathfrak{P}_{n-1}\right\}$ and $(\Psi^{n-1}_1, \ldots, \Psi^{n-1}_{K_{n-1}})$ be an orthonormal basis of $ {\rm Span} \left\{\mathfrak{P}_{n-1}\right\}$. Let $A^{n-1}:=(A^{n-1}_{kl, m})_{ 1\leq m \leq M,1\leq k,l\leq K_{n-1}} \in  \mathbb{R}^{ K_{n-1}^2 \times M} $ be defined by
    $$
    A^{n-1}_{kl,m}:= \int_\Omega \varphi_m \overline{\Psi^{n-1}_k}\Psi^{n-1}_l.
    $$

 Let $C^{n-1}:= {\rm Ker}(A^{n-1})^\perp \subset \mathbb{R}^M$ and 
    $$
    V^{n-1}:=\left\{v = \sum_{m=1}^{M} c_m \varphi_m, \; c:=(c_m)_{1\leq m \leq M} \in C^{n-1}\right\} \subset {\rm Span}\{\Phi\}.
    $$
    
Compute $v^n \in V^{n-1}$ solution to
    \begin{equation}\label{eq:dualn}
    F^n = \mathop{\max}_{
    \begin{array}{c}
    v\in V^{n-1} \\
    \forall \Psi \in {\rm Span}\{\mathfrak{P}_{n-1}\}, \quad \langle \Psi | H_{N,\Omega}^v| \Psi\rangle \geq 0. 
    \end{array}
    } \int_{\mathbb{R}^3} v \rho
    \end{equation}

\begin{remark}
Using the results of semi-definite positive programming and using similar arguments as in the proof of Theorem~\ref{thm:dual}, it can be easily checked that there exists at least one maximizer to (\ref{eq:dualn}). In addition, any maximizer to (\ref{eq:dualn}) is also a maximizer to 
$$
F^n = \mathop{\max}_{
    \begin{array}{c}
    v\in {\rm Span}\{\Phi\} \\
    \forall \Psi \in {\rm Span}\{\mathfrak{P}_{n-1}\}, \quad \langle \Psi | H_{N,\Omega}^v| \Psi\rangle \geq 0. 
    \end{array}
    } \int_{\Omega} v \rho,
$$
since $\int_{\Omega} v\rho=0$ for any $v = \sum_{m=1}^M c_m \varphi_m$ with $c:=(c_m)_{1\leq m \leq M} \in {\rm Ker}(A^{n-1})$. 
\end{remark}

{\itshape \bfseries Step~2:} Compute $\Psi_0^{v_n} \in \mathcal H_2^N(\Omega)$ a $L^2$-normalized solution to 
    \begin{equation}\label{eq:eigenvalue}
    H^{v_n}_{N,\Omega}\Psi_0^{v_n} = E(v_n)\Psi_0^{v_n},
    \end{equation}
    where $E(v_n)$ is the smallest eigenvalue of $ H^{v_n}_{N,\Omega}$.

\medskip

{\itshape \bfseries Step~3:} We now distinguish two different cases.

  \begin{itemize}    
\item {\bf Case 1: } $E(v_n)< 0$

Define $\widetilde{\mathfrak{P}}_n:= \mathfrak{P}_{n-1} \cup \{\Psi_0^{v_n} \}$. Let $K_n:= {\rm dim}\; {\rm Span}\{\widetilde{{\mathfrak{P}}}_n\}$ and let $\widetilde{\Psi}_1^n, \ldots, 
\widetilde{\Psi}_{K_n}^n$ be an orthonormal basis of ${\rm Span}\{\widetilde{{\mathfrak{P}}}_n\}$.
  
   Compute $\widetilde{S}^n:=(\widetilde{S}^n_{kl})_{1\leq k ,l \leq K_N}\in \mathcal S_+^{K_n}$ solution to 
\begin{equation}\label{eq:primal}
    \widetilde{F}^n:=\mathop{\min}_{
   \begin{array}{c}
   (S_{kl})_{1\leq k,l \leq K_n}\in \mathcal{S}_+^{K_n},\\
   \forall 1\leq m \leq M, \\
\sum_{k,l=1}^{K_n} S_{kl} \int_{\Omega} \varphi_m \overline{\widetilde{\Psi}_k^n}\widetilde{\Psi}_l^n=  \rho^m\\
   \end{array}
   } \sum_{k,l=1}^{K_n} S_{kl} \langle \widetilde{\Psi}_k^n | H_{N,\Omega} |\widetilde{\Psi}_l^n\rangle
   \end{equation}
   
Then, it holds that $\widetilde{S}^n = \sum_{k=1}^{K_n} \omega_k^n (U_k^n) (U_k^n)^T$
where $(\omega_k^n)_{1\leq k \leq K_n}\in \mathbb{R_+}^{K_n}$ are the 
eigenvalues of $\widetilde{S}^n$ (assumed to be ranked in non-increasing
order) and for all $1\leq k \leq K_n$, $U_k^n:= (U_{kl}^n)_{1\leq l \leq 
K_n}\in \mathbb{R}^{K_n}$ is a normalized eigenvector associated with 
$\omega_k^n$ so that $(U_1^n, \ldots, U_{K_n}^n)$ forms an orthonormal 
basis of $\mathbb{R}^{K_n}$. 

\medskip

Let $\mathcal K_n:= \max \left\{ k\in \{1, \ldots, K_n\}, \omega_k^n > 0\right\}$. For all $1\leq k \leq \mathcal K_n$, let 
$
\Psi_k^n:= \sum_{l=1}^{K_n} U_{kl}^n \widetilde{\Psi}_l^n,
$
and $S^n:= {\rm diag}(\omega_1^n, \ldots, \omega_{\mathcal K_n}^n)\in \mathcal S_{+,*}^{\mathcal K_n}$. We then denote by $\mathfrak{P}_n:= \left\{ \Psi_1^n, \ldots, \Psi_{{\mathcal K}_n}^n\right\}$.
   
 Define $n:=n+1$ and proceed with the next iteration.

\item {\bf Case 2: }$E(v_n)\geq 0$

Stop the algorithm.

\end{itemize}

\subsection{Property of the MCAL iterative scheme}

\medskip

We prove the following lemma, which states that the sequence of approximations yielded by the MCAL algorithm is non-increasing. Note however that we do not prove here that the sequence converges indeed to $F_{L,\Omega}^\Phi[\rho]$.

\begin{lemma}
    For all $n\geq 1$, it holds that
    $$
    F^n \geq \widetilde{F}^n = F^{n+1} \geq F_{L,\Omega}^\Phi[\rho]. 
    $$
\end{lemma}

\begin{proof}
    The first inequality is simple to see since the dual problem associated to (\ref{eq:primal}) is 
    \begin{align*}
    \widetilde{F}_n & = \mathop{\max}_{
    \begin{array}{c}
     v\in {\rm Span}\{\Phi\} \\
     \forall \Psi \in {\rm Span}\left\{\widetilde{\mathfrak{P}}_{n}\right\}, \quad \langle \Psi | H_{N,\Omega}^v| \Psi\rangle \geq 0. 
    \end{array}
    } \int_{\mathbb{R}^3} v \rho\\
    & \leq \mathop{\max}_{
    \begin{array}{c}
   \scriptsize v\in {\rm Span}\{\Phi\} \\
    \scriptsize \forall \Psi \in {\rm Span}\left\{\mathfrak{P}_{n-1}\right\}, \quad \langle \Psi | H_{N,\Omega}^v| \Psi\rangle \geq 0. 
    \end{array}
    } \int_{\mathbb{R}^3} v \rho \\
    & = F^n,\\
    \end{align*}
    since $\mathfrak{P}_{n-1} \subset \widetilde{\mathfrak{P}}_n$. 
    
    The second equality comes from the fact that
    \begin{align*}
    \widetilde{F}_n & = \mathop{\min}_{
   \begin{array}{c}
   (S_{kl})_{1\leq k,l\leq K_n}\in \mathcal S_+^{K_n},\\
   \forall 1\leq m \leq M, \\
\sum_{k,l=1}^{K_n} S_{kl} \int_{\Omega} \varphi_m \overline{\widetilde{\Psi}_k^n} \widetilde{\Psi}_l^n=  \rho^m\\
   \end{array}
   } \sum_{k,l=1}^{K_n} S_{kl} \langle \widetilde{\Psi}_k^n | H_{N,\Omega} |\widetilde{\Psi}_l^n\rangle\\
   & = \mathop{\min}_{
   \begin{array}{c}
   (S_{kl})_{1\leq k,l\leq {\mathcal K}_n}\in \mathcal S_+^{{\mathcal K}_n},\\
   \forall 1\leq m \leq M, \\
\sum_{k,l=1}^{{\mathcal K}_n} S_{kl} \int_{\Omega} \varphi_m \overline{\Psi_k^n} \Psi_l^n=  \rho^m\\
   \end{array}
   } \sum_{k,l=1}^{{\mathcal K}_n} S_{kl} \langle \Psi_k^n | H_{N,\Omega} |\Psi_l^n\rangle.\\
    \end{align*}
    In addition, we know, by definition of ${\mathcal K}_n$ and of $\Psi_1^n$, ..., $\Psi_{{\mathcal K}_n}^n$ that there exists at least one minimizer to the second minimization problem which is a positive definite matrix, that is the diagonal matrix with entries $\omega_1^n, \ldots, \omega_{ {\mathcal K}_n}$. 
  Using standard results of semi-definite positive programming, it holds that the dual problem associated to the second minimization problem introduced in the last line of the calculations above is precisely 
    $$
     F^{n+1} = \mathop{\max}_{
    \begin{array}{c}
    v\in {\rm Span}\{\Phi\} \\
    \forall \Psi \in {\rm Span}\left\{\mathfrak{P}_{n}\right\}, \quad \langle \Psi | H_{N,\Omega}^v| \Psi\rangle \geq 0.
    \end{array}
    } \int_{\mathbb{R}^3} v \rho = \widetilde{F}_n.
    $$
    Hence the desired result. 
\end{proof}

\subsection{Numerical results}

The numerical tests presented in this section were performed using Julia. In particular, the finite element code developped in~\cite{quan2023finite} was used to solve the eigenvalue problems (\ref{eq:eigenvalue}), and the ProxSDP library was used for the resolution of the semi-definite programming problems. The associated code can be found on ZENODO with the DOI 10.5281/zenodo.11669900. 

\medskip

We present in this section some preliminary numerical results on a toy numerical test case with $N = 2$, $d = 1$ and $\Omega = (-L, L)$ with $L = 10$. More precisely, for a given value $D\in \mathbb{N}^*$, the solution of problems (\ref{eq:eigenvalue}) is approximated using a Galerkin approximation in the finite element ($\mathbb{P}_1$) discretization space
$$
W^D = {\rm Span}\left\{ \phi^1 \wedge \phi^2, \phi^1, \phi^2 \in V^D\right\}
$$
where 
$$
V^D = \left\{ \phi \in \mathcal C(\Omega)| \quad \phi(-L) = \phi(L) = 0, \;  \phi|_{\left(-L + \frac{(i-1)2L}{D}, -L +\frac{i 2L}{D}\right)} \in \mathbb{P}_1, \; \forall 0\leq i \leq D \right\},
$$
and 
$$
\forall \phi^1, \phi^2\in V^D, \; \forall x,y\in \Omega, \quad \phi^1 \wedge \phi^2(x,y) = \frac{1}{\sqrt{2}}(\phi^1(x)\phi^2(y) - \phi^2(x)\phi^1(y)).
$$
The moment functions $(\varphi_m)_{1\leq m \leq M}$ are chosen to be $\mathbb{P}_1$ hat functions associated to a uniform discretization of $\Omega$ so that
$$
Z^M:={\rm Span}\left\{ \varphi_m, \; 1\leq m \leq M\right\} = \left\{ v \in \mathcal C(\Omega)| \quad\;  v|_{\left(-L + \frac{(j-1)2L}{M-1}, -L +\frac{j 2L}{M-1}\right)} \in \mathbb{P}_1, \; \forall 0\leq j \leq M-1 \right\}.
$$

The electronic density $\rho$ of choice is constructed as follows: we define, for all $x\in \Omega$, 
$$
\phi_{\rm even}(x)= 1 - \frac{|x|}{L}  \; \mbox{ and } \; \phi_{\rm odd}(x)= 
\begin{cases}
 1 - |2x+L|/L & \mbox{ if } x\leq 0,\\
 \frac{|2x-L|}{L} -1 & \mbox{ otherwise}.\\
\end{cases}
$$
Then, we define $\widetilde{\Psi}^0_1 = \frac{\phi_{\rm even}}{\|\phi_{\rm even}\|_{L^2(\Omega)}} \wedge \frac{\phi_{\rm odd}}{\|\phi_{\rm odd}\|_{L^2(\Omega)}}$ and 
$\rho = \frac{1}{2}\left( \frac{|\phi_{\rm even}|^2}{\|\phi_{\rm even}\|_{L^2(\Omega)}^2} + \frac{|\phi_{\rm odd}|^2}{\|\phi_{\rm odd}\|_{L^2(\Omega)}^2}\right).$

We then apply the MCAL algorithm starting from $\widetilde{\mathfrak{P_0}} = \left\{\widetilde{\Psi}^0_1\right\}$.

\medskip

Let us first highlight the influence of the parameter $q_{\rm vec}$ on the performance of the algorithm in terms of the number of iterations required to achieve numerical convergence. We first conduct a first series of tests with $M = 20$ and $D = 100$. 

Figure~\ref{fig:qvec} highlights the behaviour of the numerical scheme with respect to $q_{\rm vec}$.

\begin{figure}[h!]
\centering
\includegraphics[width=12cm]{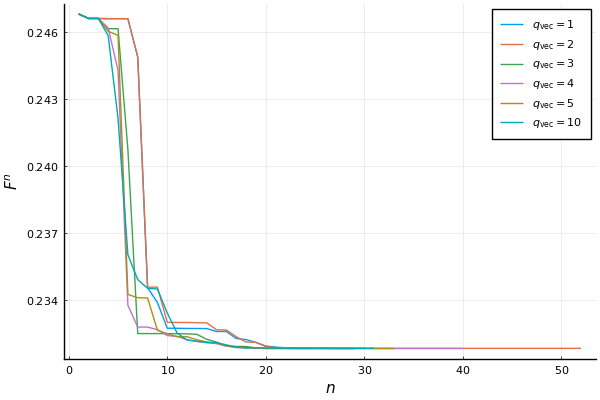}
\includegraphics[width=12cm]{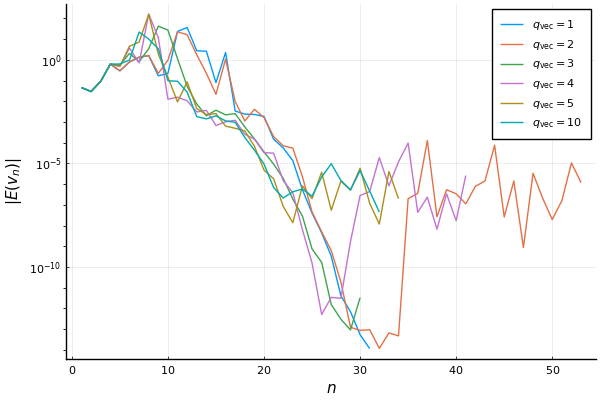}
\caption{Evolution of $F^n$ (above) and $|E(v_n)|$ (below) as a function of $n$ for different values of $M$.}\label{fig:qvec}
\end{figure}

The upper (respectively lower) figure shows the values of $F^n$ (respectively $|E(v_n)|$) as a function of $n$ for different values of $M$. As predicted by our theoretical results, for any value of $q_{\rm vec}$, the sequence $(F^n)_{n\in\mathbb{N}}$ is non-increasing and we also checked numerically that $\widetilde{F}^n = F^{n+1}$ for all $n\in \mathbb{N}$. In constrast, the sequence $(E(v_n))_{n\in\mathbb{N}^*}$ is not monotonous. We also observe that for any tested value of $q_{\rm vec}$, the sequence $(F^n)_{n\in \mathbb{N}}$ converges to the same limit value. It seems that for greater values of $q_{\rm vec}$, the number of iterations $n$ needed for the algorithm to converge is lower. 

\medskip

Figure~\ref{fig:M} highlights the behaviour of the numerical scheme with respect to the number $M$ of moment constraints. In these tests, $D = 100$ and $q_{\rm vec} = 4$.

Again, the upper (respectively lower) figure shows the values of $F^n$ (respectively $|E(v_n)|$) as a function of $n$ for different values of $M$. As before, we observe that the sequence $(F^n)_{n\in\mathbb{N}}$ is non-increasing and we also checked numerically that $\widetilde{F}^n = F^{n+1}$ for all $n\in \mathbb{N}$. In constrast, the sequence $(E(v_n))_{n\in\mathbb{N}^*}$ is not monotonous. We also observe that for any tested value of $M$, the sequence $(F^n)_{n\in \mathbb{N}}$ converges to some limit value denoted here by $F^\infty(M)$ which depends on $M$. We observe again that the value of $F^\infty(M)$ does not increase monotonically with $M$, which stems from the fact that the spaces $Z^M$ do not form an increasing family of vector spaces for the inclusion. However, is still holds that $Z^{10} \subset Z^{20} \subset Z^{40}$, and we indeed observe that $F^\infty(10) \leq F^\infty(20) \leq F^\infty(40)$, which is coherent with the variational structure of the moment constraint approach studied here. We also observe that the value of $E(v_n)$ seems to stagnate in most of the numerical tests (except the one corresponding to $M=10$) to a value close to $-10^{-5}$.  

\begin{figure}[h!]
\centering
\includegraphics[width=12cm]{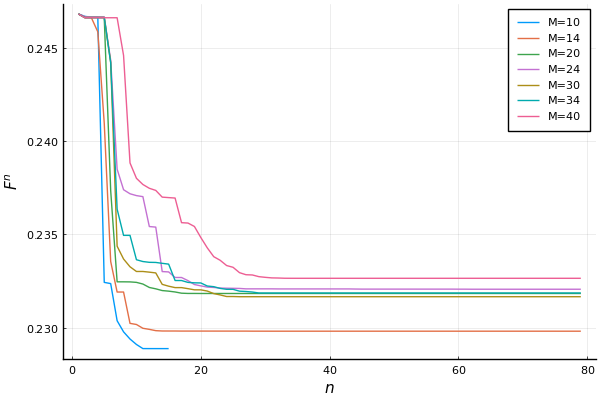}
\includegraphics[width=12cm]{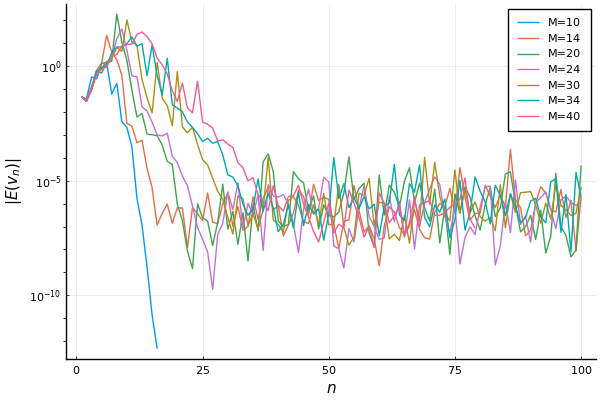}
\caption{Evolution of $F^n$ (above) and $|E(v_n)|$ (below) as a function of $n$ for different values of $M$.}\label{fig:M}
\end{figure}

Lastly, Figure~\ref{fig:potential} shows the plots of the potential $v_n$ obtained after running $n = 80$ iterations of the MCAL algorithm for various values of $M$ ($M=10, 20, 30, 40$). We observe that the potential value seems to converge to some limit value of $M$ increases. However, the number of moment constraints should definitely be higher to obtain a better accuracy, which was not possible with our current implementation. More evolved versions of the present MCAL algorothm should be designed to alleviate this bottleneck, which will be the object of a future work.

\begin{figure}[h!]
\centering
\includegraphics[width=6.7cm]{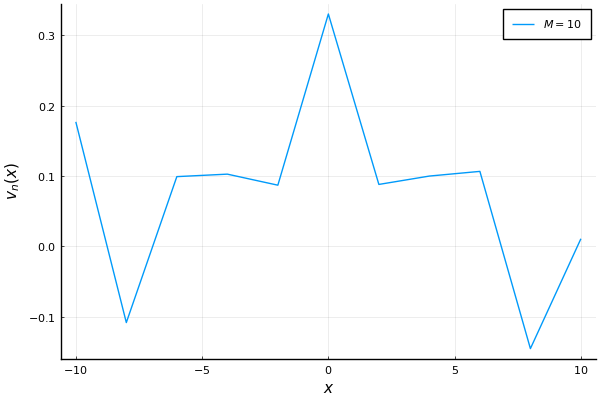}
\includegraphics[width=6.7cm]{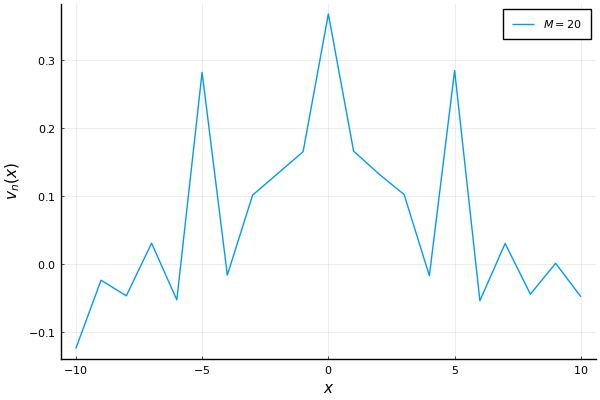}
\includegraphics[width=6.7cm]{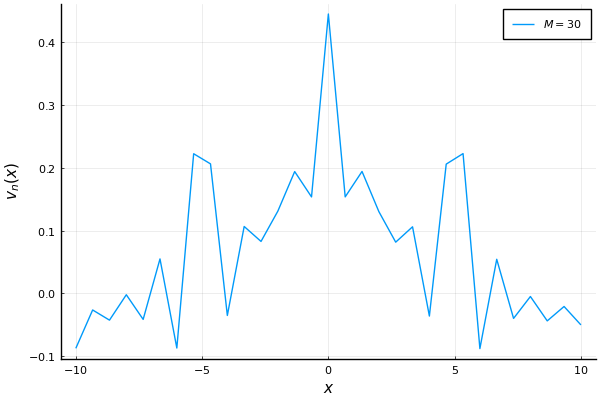}
\includegraphics[width=6.7cm]{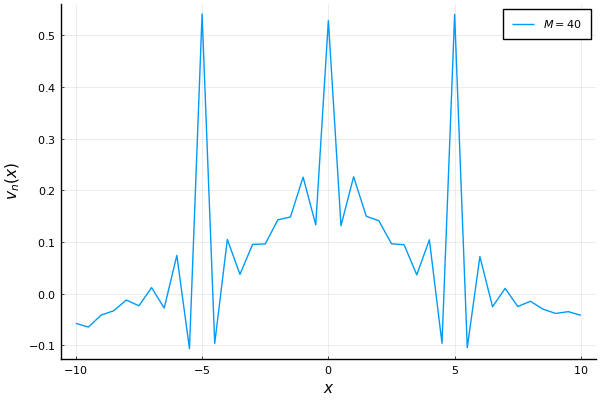}
\caption{$v_n(x)$ as a function of $x$ for $M=10,20,30,40$ ($n=80$)}\label{fig:potential}
\end{figure}

\section{Proofs}

We gather in this section the proofs of our main theoretical results.

\subsection{Proof of Theorem~\ref{thm:existence}}\label{sec:proofthm6}

\begin{proof}[Proof of Theorem~\ref{thm:existence}]
{\bfseries Step~1: (Finiteness)} Since $\rho \in \mathcal I_N$, there exists at least one element $\Psi_0 \in \mathcal H_1^N$ such that $\rho_{\Psi_0} = \rho$. 
Denoting by $\Gamma_0 := |\Psi_0\rangle \langle \Psi_0|$, it can then be easily seen that $\Gamma_0 \in \mathfrak{S}_1^+(\mathcal H_0^N,\Phi,\rho)$ and that $\tr(\Theta\Gamma_0) = \int_{\mathbb{R}^3} \theta(|x|)\rho(x)\,dx = C_\rho$. Thus, we immediately obtain that for all $C\geq C_\rho$, $F_{L,\theta}^{\Phi, C}[\rho] > -\infty$.  

{\bfseries Step~2: (Existence of minimizer)} Let $(\Gamma_n)_{n\in\mathbb{N}}$ be a minimizing sequence associated to  \eqref{min:LiebMomentA}. Then, we know from the proof of Theorem~4.4 of \cite{Lieb-83b} that, up to the extraction of a subsequence, there exists a trace-class operator $\Gamma_\infty \in \mathfrak{S}_1^+(\mathcal H_0^N)$  such that $\left((H_N+D)^{1/2}\Gamma_n(H_N+D)^{1/2}\right)_{n\in\mathbb{N}}$ weakly converges in the sense of trace-class operators to $(H_N+D)^{1/2}\Gamma_\infty(H_N+D)^{1/2}$ as $n$ goes to infinity. To prove that $\Gamma_\infty$ is a minimizer to \eqref{min:LiebMomentA}, it is sufficient to prove that $\rho_{\Gamma_{\infty}}$ satisfies
$$
\forall 1\leq m \leq M, \; \int_{\mathbb{R}^3} \rho_{\Gamma_\infty}\varphi_m = \int_{\mathbb{R}^3} \rho \varphi_m \; \mbox{ and }\; \int_{\mathbb{R}^3} \rho_{\Gamma_\infty}(x) \theta(|x|)\,dx = \tr( \Theta \Gamma_\infty)\leq C.
$$
For all $n\in \mathbb{N}$, let us denote by $\tau_n \in L^2(\mathbb{R}^{3N} \times \mathbb{R}^{3N})$ the kernel of $\Gamma_n$ and by $\tau_\infty \in L^2(\mathbb{R}^{3N} \times \mathbb{R}^{3N})$ the kernel of $\Gamma_\infty$. Let us also denote for all $n\in \mathbb{N}$, 
$$
\gamma_n (x_1, \ldots, x_N) := \tau_n(x_1, \ldots, x_N; x_1, \ldots, x_N)
$$ and by 
$$
\gamma_\infty (x_1, \ldots, x_N) := \tau_\infty(x_1, \ldots, x_N; x_1, \ldots, x_N)
$$ for all $x_1,\ldots,x_N\in \mathbb{R}^3$. Let us prove that $(\gamma_n)_{n\in\mathbb{N}}$ is a tight sequence. Indeed, let $R>0$ and $B_R$ be the ball of radius $R$ of $\mathbb{R}^{3N}$. Then, denoting by $\un_{B^c_R}$ the characteristic function of the set $B_R^c$, it holds that for all $n\in \mathbb{N}$, 
\begin{align*}
    \int_{B^c_R} \gamma_n & = \int_{\mathbb{R}^{3N}} \un_{B^c_R}\gamma_n\\
    & \leq \int_{\mathbb{R}^{3N}} \left(\frac{1}{N}\sum_{i=1}^N \frac{\theta(|x_i|)}{\theta(R)}\right)\gamma_n(x_1,\ldots,x_N)\,dx_1\ldots\,dx_N\\
    & = \frac{1}{\theta(R)}\tr (\Theta\Gamma_n) \leq \frac{C}{\theta(R)}.\\
\end{align*}
Let us denote by $M_P$ the multiplication operator by any function $P$ bounded with compact support on $\R^{3N}$. We then know from the proof of Theorem~4.4 of~\cite{Lieb-83b} that
$$
{\rm Tr}(M_P \Gamma^\infty) = \mathop{\lim}_{n\to +\infty}{\rm Tr}(M_P \Gamma^n). 
$$
This, together with the tightness result above, yields that  $(\rho_{\Gamma_n})_{n\in\mathbb{N}}$ weakly converges to $\rho_{\Gamma_\infty}$ in $L^1(\mathbb{R}^3)$. It then easily follows that for all $m=1,\cdots,M,$
\[ \int_{\R^3}\varphi_m\rho_{\Gamma_\infty}=\lim_{n\to+\infty}\int_{\R^3}\varphi_m\rho_{\Gamma_n}=\int_{\R^3}\varphi_m\rho\]
and that
$$
 \int_{\R^3}\theta(|x|)\rho_{\Gamma_\infty}(x)\,dx = \tr(\Theta \Gamma_\infty) \leq C.
$$
The operator $\Gamma_\infty$ is thus a minimizer of (\ref{min:LiebMomentA}). In particular, since $\un \in {\rm Span}\{\Phi\}$, it holds that $\tr(\Gamma_{\infty}) = N$. 

{\bfseries Step~3: (Existence of a sparse minimizer)} 

Let us now introduce the function $\Lambda : \mathcal H_1^N\to\R^{M+1}$ such that for all $m=1,\cdots,M,$
    \[\Lambda_m(\Psi)=\int_{\R^3}\varphi_m(x)\rho_\Psi(x)\dd x=\int_{\R^{dN}}\varphi_m(x)|\Psi(x,x_2,...,x_N)|^2\dd x\dd x_2...\dd x_N,\]
    and
    \[\Lambda_{M+1}(\Psi)=\langle \Psi|H_N|\Psi\rangle.\] It can then be easily seen that $\Lambda$ is a continuous map on $\mathcal H_1^N$.

Let $\Gamma_{\rm min}$ be a minimizer of (\ref{min:LiebMomentA}). Then, there exists a countable index set $\mathcal J \subset \mathbb{N}$, an orthonormal family $(\Psi_j)_{j\in \mathcal J}$ of $\mathcal H_0^N$ and a family of positive numbers $(\alpha_j)_{j\in\mathcal J}$ such that $\sum_{j \in \mathcal J} \alpha_j = N$ (this comes from the fact that $\un \in {\rm Span}\{\Phi\}$) and 
$$
\Gamma_{\rm min} = \sum_{j \in \mathcal J} \alpha_j |\Psi_j\rangle \langle \Psi_j|. 
$$
In addition, it can be easily checked that $\Psi_j\in \mathcal H_1^N$ for all $j\in \mathcal J$. 
We then define $\mu_{\rm min}:= \sum_{j\in \mathcal J} \alpha_j \delta_{\Psi_j}$ which is a Borel measure on $\mathcal B(\mathcal H_1^N)$ since $\tr(H_N \Gamma_{\rm min})$ is finite and $\tr \Gamma_{\rm min}= N$. It can then be easily checked that 
$$
\int_{\mathcal H_1^N} \|\Lambda(\Psi)\| \,d\mu_{\rm min}(\Psi) < +\infty. 
$$
Thus, by Proposition~\ref{prop:tchakaloff}, there exist $1 \leq K\leq M+1$, $\overline{\Psi}_1,\cdots,\overline{\Psi}_K\in\mathcal H_1^N$ and $\omega_1,\cdots,\omega_K>0$ such that
    \[\int_{\mathcal H_1^N}\Lambda(\Psi)\dd\mu_{\rm min}(\Psi)=\sum_{k=1}^K\omega_k\Lambda(\overline{\Psi}_k).\]
    
Denoting by $\Gamma_K=\sum_{k=1}^K\omega_k|\overline{\Psi}_k\rangle\langle\overline{\Psi}_k|$, it can then be easily checked that $\Gamma_K$ is also a minimizer to (\ref{min:LiebMomentA}). Hence the desired result.

\end{proof}

\subsection{Proof of Theorem~\ref{thm:convergence}}\label{sec:convergencethm}

\begin{proof}[Proof of Theorem~\ref{thm:convergence}]
The first assertion of the theorem is a direct consequence of Theorem~\ref{thm:existence}. 
Using the same arguments as in the proof of Theorem~\ref{thm:existence}, one can easily obtain that the sequence $\left( (H_N+ D)^{1/2} \Gamma_n (H_N+ D)^{1/2} \right)_{n\geq n_0}$ is compact in $\mathfrak{S}_1^+(\mathcal H_0^N)$. Thus, up to the extraction of a subsequence there exists $\Gamma_\infty \in \mathfrak{S}_1^+(\mathcal H_0^N)$ such that $\tr(H_N \Gamma_\infty) <+\infty$ and such that $\left( (H_N+ D)^{1/2} \Gamma_n (H_N+ D)^{1/2} \right)_{n\geq n_0}$ weakly converges to $\left( (H_N+ D)^{1/2} \Gamma_\infty (H_N+ D)^{1/2} \right)_{n\geq n_0}$ in the sense of trace-class operators of $\mathfrak{S}_1^+(\mathcal H_0^N)$. 

Moreover, following again the same lines of proof, we obtain that the sequence $(\rho_{\Gamma_n})_{n\geq n_0}$ weakly converges in $L^1(\mathbb{R}^3)$ to $\rho_{\Gamma_\infty}$. As a consequence, it holds that $\int_{\mathbb{R}^3} \theta(|x|)\rho_{\Gamma_\infty}(x)\,dx \leq C$. Moreover, since for all $n\in\N^*$, $\int_{\mathbb{R}^3} \varphi^n_m \rho_{\Gamma_n} =\int_{\mathbb{R}^3} \varphi^n_m \rho$,  using Lemma~\ref{lem:useful}, we then obtain that, necessarily, $\rho_{\Gamma_\infty} = \rho$. This makes $\Gamma_\infty$ admissible for \eqref{min:lieb} so that we have that $\tr(H_N\Gamma_\infty)\geq F_L[\rho]$.
    Notice now that for all $n\geq n_0$, $-\infty <F^{\Phi^n,C}_{L,\theta}[\rho] \leq F_L[\rho]$. Thus for any converging subsequence of $(F^{\Phi^n,C}_{L,\theta}[\rho])_{n\geq n_0}$ to some limit $F_L^\infty$, it holds that $-\infty < F_L^\infty\leq F_L[\rho]$. For this subsequence, still denoted by $(F^{\Phi^n,C}_{L,\theta}[\rho])_{n\geq n_0}$ for the sake of simplicity, it holds that $\mathop{\lim}_{n\to\infty}\tr(H_N\Gamma_n)=F_L^{\infty}$, and we then have that
    \[F_{L}[\rho]\geq F^{\infty}_L\geq \tr(H_N\Gamma_\infty).\]
    Thus, necessarily, $\Gamma_\infty$ is a minimizer of (\ref{min:lieb}). Moreover, $F_L^\infty = F_L[\rho]$ for any extracted subsequence so that $\displaystyle \mathop{\lim}_{n\to +\infty}{\rm Tr}(H_N \Gamma_n) = {\rm Tr}(H_N \Gamma_\infty)$. Using the compactness of the Fock space of bounded particle number for the geometric convergence~\cite{lewin2011geometric}[Lemma~2.2, Lemma~2.3], we thus obtain the desired result.
\end{proof}

\subsection{Proof of Theorem~\ref{thm:dual}}\label{sec:thmdual}
\begin{proof}[Proof of Theorem~\ref{thm:dual}]
{\bfseries Step~1:} Let us first prove that there exists a maximizer to the optimization problem
\begin{equation}\label{eq:DiscreteProblem}
 \mathop{\sup}_{
\begin{array}{c}
v \in {\rm Span}\{\Phi\},\\
\forall \Psi \in {\rm Span}\{\Psi_1, \ldots, \Psi_K\}, \quad \langle\Psi | H_{N,\Omega}^v|\Psi \rangle \geq 0 \\
\end{array}} \int_\Omega v \rho.
\end{equation}

We denote here by $\mathcal S^K$ the set of symmetric matrices of $\mathbb{R}^{K\times K}$. For any $\varphi \in {\rm Span}\{\Phi\}$, let us consider the linear form $l_\varphi: \mathcal S^K \to \mathbb{R}$ defined as follows: 
$$
\forall S:=(S_{kl})_{1\leq k,l\leq K}\in\mathcal S^K, \quad l_\varphi(S):= \int_{\Omega} \varphi(x) \sum_{k,l=1}^K S_{kl} \overline{\Psi_k}(x)\Psi_l(x)\,dx = {\rm Tr}(\varphi \Gamma_S), 
$$
where 
$$
\Gamma_S := \sum_{k,l=1}^K S_{kl} |\Psi_k \rangle \langle \Psi_l|.
$$
Let us now consider the vectorial space
$$
L:=\{ l_\varphi, \; \varphi\in {\rm Span}\{\Phi\}\}.
$$
The space $L$ is a finite-dimensional subspace of the set of linear forms on $\mathcal S^K$, and its dimension $J$ is lower or equal to the dimension of ${\rm Span}\{\Phi\}$. Let $(\widetilde{l}_1, \ldots, \widetilde{l}_J)$ be a basis of $L$. By construction, there exists $\widetilde{\varphi}_1, \ldots, \widetilde{\varphi}_J \in {\rm Span}\{\Phi\}$ such that $\widetilde{l}_j = l_{\widetilde{\varphi}_j}$ for all $1\leq j \leq J$. Let us then denote by $\widetilde{\Phi}:= \{\widetilde{\varphi}_1, \ldots, \widetilde{\varphi}_J\}$. It can then be easily checked that any element $\varphi$ of ${\rm Span}\{\Phi\}$ can be rewritten as 
$$
\varphi = \widetilde{\varphi} + \varphi_0,
$$
where $\widetilde{\varphi} \in {\rm Span}\{\widetilde{\Phi}\}$ and $\varphi_0 \in {\rm Span}\{\Phi\}$ such that $l_{\varphi_0} = 0$. In particular, this implies that $\int_\Omega \varphi_0 \rho = 0$ since for all $\varphi \in {\rm Span}\{\Phi\}$,
$$
\int_\Omega \varphi \rho = \int_\Omega \varphi(x) \sum_{k=1}^K \omega_k|\Psi_k(x)|^2\,dx = l_\varphi( {\rm diag}(\omega_1, \ldots, \omega_K)).
$$
Thus, proving that there exists a maximizer to (\ref{eq:DiscreteProblem}) is equivalent to proving that there exists a maximizer to 
\begin{equation}\label{eq:DiscreteProblem2}
 \mathop{\sup}_{
\begin{array}{c}
v \in {\rm Span}\{\widetilde{\Phi}\},\\
\forall \Psi \in {\rm Span}\{\Psi_1, \ldots, \Psi_K\}, \quad \langle\Psi | H_{N,\Omega}^v|\Psi \rangle \geq 0 \\
\end{array}} \int_\Omega v \rho.
\end{equation}
Now, by definition of $\widetilde{\varphi}_1, \ldots, \widetilde{\varphi}_J$, it holds that the application 
$A: \mathcal S^K \to \mathbb{R}^J$ defined so that for all $1\leq j \leq J$ and all $S = (S_{kl})_{1\leq k,l\leq K}$, 
$$
A(S)_j:= \int_\Omega \widetilde{\varphi}_j \sum_{k,l=1}^K S_{kl} \overline{\Psi_k}\Psi_l
$$
is surjective. Indeed, this comes from the fact that ${\rm dim}\; {\rm Rank}(A) = {\rm dim}\; L = J$. It can then be easily checked that (\ref{eq:DiscreteProblem2}) is then equivalent to the dual semi-definite programming problem: 
\begin{equation}\label{eq:DP}
 \mathop{\sup}_{
\begin{array}{c}
(y,S) \in \mathbb{R}^J \times \mathcal S^K\\
A^*(y) + S = C \\
S \succcurlyeq 0\\
\end{array}} \langle b,y\rangle,
\end{equation}
where $b = (b_j)_{1\leq j \leq J}$ is such that $b_j = \int_\Omega \widetilde{\varphi}_j \rho$ for all $1\leq j \leq J$ and $C = (C_{kl})_{1\leq k,l \leq K} \in \mathcal S^K$ with 
$$
C_{kl}:= \langle \Psi_k | H_{N,\Omega} | \Psi_l \rangle \quad \forall 1\leq k,l \leq K.
$$
Indeed, if $(y,S)\in \mathbb{R}^J \times \mathcal S^K$ is a maximizer to (\ref{eq:DP}), it holds that $v = \sum_{j=1}^J y_j \widetilde{\varphi}_j$ is a maximizer to (\ref{eq:DiscreteProblem2}), and thus to (\ref{eq:DiscreteProblem}).

The primal problem associated to \eqref{eq:DP} reads as
\begin{equation}\label{eq:Dprim}
 \mathop{\inf}_{
\begin{array}{c}
X \in \mathcal S^K\\
A(X)  = b \\
X \succcurlyeq 0\\
\end{array}} \langle C,X\rangle,
\end{equation}

Let us also remark that $\int_\Omega \rho \varphi = \int_\Omega \rho_{\Gamma_S} \varphi$ for all $\varphi\in {\rm Span}\{\Phi\}$ if and only if $A(S) = b$. Thus, this implies that there exists at least one minimizer $X$ to (\ref{eq:Dprim}) which is given by $X = {\rm diag}(\omega_1, \ldots, \omega_K)$ and  is positive definite. Using Theorem~\ref{th:SDPex}, we then obtain the existence of at least one maximizer to (\ref{eq:DP}), and hence to (\ref{eq:DiscreteProblem}) and (\ref{eq:DiscreteProblem2}).

\medskip

{\bfseries Step~2:} To conclude the proof of the desired result, it only remains to show that
\begin{align*}
D_{L,\Omega}^{\Phi}[\rho] & := \mathop{\sup}_{
\begin{array}{c}
v \in {\rm Span}\{\Phi\},\\
\forall \Psi\in \mathcal H_1^N(\Omega), \quad \langle\Psi | H_{N,\Omega}^v|\Psi \rangle \geq 0 \\
\end{array}} \int_\Omega v \rho \\
& = \mathop{\sup}_{
\begin{array}{c}
v \in {\rm Span}\{\Phi\},\\
\forall \Psi \in {\rm Span}\{\Psi_1, \ldots, \Psi_K\}, \quad \langle\Psi| H_{N,\Omega}^v|\Psi \rangle \geq 0 \\
\end{array}} \int_\Omega v \rho.\\
\end{align*}
On the one hand, it holds from Theorem~\ref{thm:stronDuality}, that $D_{L,\Omega}^\Phi[\rho] = F_{L,\Omega}^\Phi[\rho]$. On the other hand, using similar arguments as in the proof of Theorem~\ref{thm:stronDuality}, it holds that 
$$
\mathop{\sup}_{
\begin{array}{c}
v \in {\rm Span}\{\Phi\},\\
\forall \Psi \in \mathcal W, \quad \langle\Psi | H_{N,\Omega}^v|\Psi\rangle \geq 0 \\
\end{array}} \int_\Omega v \rho = \mathop{\inf}_{\Gamma \in {\mathfrak S}_1^+(\mathcal W, \Phi,\rho)} {\rm Tr}(H_{N,\Omega} \Gamma),
$$
where $\mathcal W:= {\rm Span}\{\Psi_1, \ldots, \Psi_K\}$. Since, by definition of $\Psi_1$, ..., $\Psi_K$, it holds that $\displaystyle F_{L,\Omega}^\Phi[\rho] = \mathop{\inf}_{\Gamma \in {\mathfrak S}_1^+(\mathcal W, \Phi,\rho)} {\rm Tr}(H_{N,\Omega} \Gamma)$, we obtain the desired result.
\end{proof}

\section*{Acknowledgements}
This
publication is part of a project that has received funding from the European Research Council (ERC) under the European Union’s Horizon 2020 Research and Innovation Programme – Grant Agreement  810367
L.N. is partially on academic leave at Inria (team Matherials) for the year 2022-2023 and 2023-2024 and acknowledges the hospitality if this institution during this period. His work  benefited from the support of the FMJH Program PGMO,  from H-Code, Université Paris-Saclay and from the ANR project GOTA (ANR-23-CE46-0001).

\bibliographystyle{alpha}
\bibliography{biblio}

\end{document}